\documentclass[journal,twoside,web]{ieeecolor}
\usepackage{generic}
\usepackage{cite}
\usepackage{amsmath,amssymb,amsfonts,cases,mathrsfs,mathtools}
\usepackage{algorithmic,algorithm}
\usepackage{xcolor,graphicx}
\usepackage{hyperref}
\hypersetup{hidelinks=true}
\usepackage{textcomp}
\def\BibTeX{{\rm B\kern-.05em{\sc i\kern-.025em b}\kern-.08em
    T\kern-.1667em\lower.7ex\hbox{E}\kern-.125emX}}
\usepackage{wrapfig,lipsum}
\usepackage{pgfplots}
\usepackage{tikz, tikzscale,calc}
\usepackage{caption}
\usepackage{subcaption}
\usepackage{enumerate}
\usepackage{booktabs}
\usepackage{array}
\usepackage{verbatim} 
\usepackage{multirow}
\newtheorem{lem}{Lemma}
\newtheorem{prop}{Proposition}
\newtheorem{defin}{Definition}
\newtheorem{ass}{Assumption}
\newtheorem{thm}{Theorem}
\newtheorem{remark}{Remark}
 \newtheorem{corol}{Corollary}
\newenvironment{lemma}{\begin{lem}}{\hfill  \end{lem}}

\begin{document}
\title{Edge-based Synchronization over Signed Digraphs with Multiple Leaders}
\author{Pelin Sekercioglu, Angela Fontan, Dimos V. Dimarogonas
\thanks{This work was supported in part by the Wallenberg AI, Autonomous Systems and Software Program (WASP) funded by the Knut and Alice Wallenberg (KAW) Foundation, the Horizon Europe EIC project SymAware (101070802), the ERC LEAFHOUND Project, the Swedish Research Council (VR), and Digital Futures.}
\thanks{P. Sekercioglu, A. Fontan, and D. V. Dimarogonas are with the Department of Decision and Control Systems, KTH Royal Institute of Technology, SE-100 44 Stockholm, Sweden (e-mail: \{pelinse,angfon,dimos\}@kth.se).}}

\maketitle

\begin{abstract}
This work addresses the edge-based synchronization problem in first-order multi-agent systems containing both cooperative and antagonistic interactions with one or multiple leader groups. The presence of multiple leaders and antagonistic interactions means that the multi-agent system typically does not achieve consensus, unless specific conditions (on the number of leaders and on the signed graph) are met, in which case the agents reach a trivial form of consensus. In general, we show that the multi-agent system exhibits a more general form of synchronization, including bipartite consensus and containment. Our approach proposes a signed edge-based agreement protocol for signed networks described by signed edge-Laplacian matrices. In particular, in this work, we present new spectral properties of signed edge-Laplacian matrices containing multiple zero eigenvalues and establish global exponential stability of the synchronization errors. Moreover, we explicitly compute the equilibrium to which all edge states converge, thereby characterizing the resulting synchronization behavior. Numerical simulations validate our theoretical results.
\end{abstract} 


\section{Introduction}
\label{sec:introduction}
\IEEEPARstart{T}{he} coordination of multi-agent systems has been widely investigated, with numerous results addressing problems such as consensus under first-order, second-order, and general linear high-order dynamics \cite{WEIBOOK}. When the interaction network includes a single leader, classical leader–follower consensus ensures that all follower agents asymptotically track the leader’s state \cite{ni2010leader}. However, this framework no longer applies when multiple leaders are present in the network. In such situations, the problem is better described as containment control \cite{Cao:11}, whose objective is to drive the followers’ states toward the containment set formed by the leaders’ states.

Numerous works have investigated distributed containment control in various settings \cite{kan2015containment,cao2012distributed,ji2008containment}. However, the majority of existing studies on consensus and containment in multi-agent systems focus on purely cooperative networks, that is, scenarios in which coordination among agents is achieved exclusively through cooperative interactions. However, in many real-world scenarios, agents may exhibit antagonistic behaviors. Examples include robotic applications \cite{sebastian2022adaptive,mathavaraj2023rigid}, and social networks where agents compete \cite{altafini2012_6329411,Fontan2021Role,Fontan2021Signed} and spread disinformation \cite{csekerciouglu2024distributed}, to name just a few. %\textcolor{blue}{This work is motivated by such real-world scenarios, particularly addressing the presence of multiple influential entities, referred to as leaders, in a signed edge-based framework.} 
A representative example is found in opinion dynamics, where consumer behavior is shaped by both peer influence and external entities, such as competing marketers (leaders). These leaders, grouped into disjoint cooperative or competitive subsets, affect consumers. Consumers may oppose a disliked marketer by aligning with the other, or adopt a position between them if they trust both, favoring the one with stronger influence. 

A common and widely adopted framework for representing both cooperative and antagonistic relationships in multi-agent systems relies on \textit{signed graphs} \cite{altafini2012_6329411,Shi2019Dynamics}. Within this representation, cooperative interactions are encoded by positive edge weights, whereas antagonistic relationships are captured through negative edge weights. In such networks, rather than achieving standard consensus, the typical emergent synchronization behavior is \textit{bipartite consensus}, where agents converge to two symmetric equilibrium points, provided the signed graph is \textit{structurally balanced}. A signed graph is said to be structurally balanced if the nodes can be divided into two disjoint subsets such that agents within the same subset cooperate, while agents in different subsets interact antagonistically \cite{altafini2012_6329411}. There are various studies on bipartite consensus control \cite{altafini2012_6329411, valcher} and on bipartite containment control \cite{meng2017bipartite,zhang2020cooperative,sekercioglu2023exponential,csekerciouglu2024distributed}.

In this paper, to address all possible emerging synchronization behaviors arising in the presence of antagonistic interactions and multiple leaders, we study the edge-based synchronization of multi-agent systems interconnected over directed signed graphs, where the edge-based formulation describes the dynamics in terms of relative differences between agents. We recast the synchronization problem into one of stability of the appropriately defined synchronization errors. The main contributions are threefold: (i) we present new properties of signed edge-Laplacian matrices (hereafter: edge Laplacians) containing multiple zero eigenvalues defining the multi-agent system (Lemma~\ref{lemma_multipleleaders} and Table~\ref{tab:eigenvalues}); (ii) we characterize the eigenvectors associated with the signed edge Laplacians within the signed edge-based network formulation; and (iii) we establish exponential stability of the synchronization set and provide explicit estimates of the edge-based limit points of the agents in the signed edge-based network formulation (Theorem~\ref{prop:result2}), using a Lyapunov equation-based analysis previously proposed in \cite{TACPelin}. In contrast to \cite{csekerciouglu2024distributed,sekercioglu2023exponential,meng2017bipartite}, which consider node-based formulations, we focus on edge-based synchronization and contribute with new spectral properties of signed (edge) Laplacians. Compared to \cite{du2018edge}, which studies edge convergence in strongly connected signed digraphs, and \cite{du2019further}, which considers signed digraphs containing a spanning tree, we address the more general case of signed digraphs with multiple leaders and establish exponential stability of the set, where all synchronization errors are equal to zero.

From a technical viewpoint, our main results on exponential stability build on \cite{csekerciouglu2024distributed} and the framework introduced in \cite{DYNCON-TAC16}, where the overall dynamics are decomposed into two coupled subsystems: the evolution of a weighted average and the synchronization errors relative to that average. We generalize the latter to the setting of directed signed graphs, adopting an edge-based representation. In particular, we recast the synchronization problem as a stability problem for synchronization errors, defined as the difference between the edge states and the weighted edge average. To capture all emergent behaviors, the analysis is carried out in signed edge-based coordinates, for which we investigate structural properties of the incidence and edge Laplacian matrices. As a key technical contribution, we provide new results that characterize the rank, null space, and eigenvector properties of signed edge Laplacians associated with directed signed graphs that contain one or multiple leaders. Then, we establish exponential stability of the synchronization set building on the Lyapunov characterization for Laplacians with multiple zero eigenvalues presented in our previous work \cite{TACPelin} and we provide the explicit asymptotic values of the edge states, revealing the emerging behavior of the multi-agent system. Together, these results yield a complete characterization of edge synchronization over directed signed graphs with leaders.  

\section{Preliminaries}\label{section3}
\textit{Notation:} $\lvert \cdot \rvert$ denotes the absolute value for scalars, the Euclidean norm for vectors, and the spectral norm for matrices. $\mbox{card}(\cdot)$ indicates the cardinality of a set. $\mbox{diag}(z)$ denotes a diagonal matrix whose diagonal elements are the entries of the vector $z$. $\mathcal{N}(A)$ denotes the null space of matrix $A$. $\mathbb{R}$ is the set of real numbers and $\mathbb{R}_{\geq 0}$ the nonnegative orthant. $A>0$ ($A\ge 0$) indicates that $A$ is a positive definite (positive semidefinite) matrix. An M-matrix is a matrix whose off-diagonal entries are nonpositive, and whose eigenvalues have nonnegative real parts. $\mbox{blkdiag}(A_i)$ indicates a block diagonal matrix formed by the matrices $A_i$.

%\subsection{Signed digraphs}\label{section3a} 
Let $\mathcal{G}_s = (\mathcal{V}, \mathcal{E})$ be a \textit{signed} graph, where $\mathcal{V} = \{ \nu_1, \nu_2, \dots, \nu_N \}$ is the set of $N$ nodes and $\mathcal{E} \subseteq \mathcal{V} \times \mathcal{V}$ is the set of edges. Each edge in $\mathcal{E}$ has a sign, either positive or negative. If all edges have a positive sign, the graph is called an \textit{unsigned} graph (in standard graph theory). The graph is undirected if information flow between agents is bidirectional, meaning $( \nu_i, \nu_j ) = ( \nu_j, \nu_i ) \in \mathcal{E}$. Otherwise, the graph is directed and is commonly referred to as a digraph. The edge $\varepsilon_k = ( \nu_j, \nu_i ) \in \mathcal{E}$ of a digraph denotes that the agent $\nu_i$, which is the terminal node (tail of the edge), can obtain information from the agent $\nu_j$, which is the initial node (head of the edge). The adjacency matrix of $\mathcal{G}_s$ is $A := [a_{ij}] \in \mathbb{R}^{N \times N}$, where $a_{ij}\ne 0$ if and only if $(\nu_j,\nu_i)\in \mathcal{E}$. $a_{ij} > 0$ if and only if the edge $(\nu_j, \nu_i)$ has a positive sign, indicating a cooperative relationship, and $a_{ij} < 0$ if and only if the edge $(\nu_j, \nu_i)$ has a negative sign, indicating an antagonistic relationship. In this work, we only consider unweighted digraphs, such that $a_{ij} = \{0, 1, -1\}$, without any self-loops. A signed digraph is said to be digon sign-symmetric if $a_{ij} a_{ji} \ge 0$. It means that the interaction between two interconnected agents has always the same sign in both directions. Throughout this work, we make the following standing assumption:

\begin{ass}\label{standing_ass}%[Standing assumption]
    The signed digraph is unweighted and digon sign-symmetric.
\end{ass}

A directed path is a sequence of distinct adjacent nodes in a digraph. When the nodes of the path are distinct except for its end vertices, the directed path is called a directed cycle. A directed spanning tree is a directed tree subgraph that includes all the nodes of the digraph. In this structure, every agent (node) has a parent node, except for the \textit{root node}, which has no incoming edges and is connected to every other node via directed paths. Since each edge points from the root to other nodes, the tree contains no cycles. A digraph is said to be strongly connected if there exists a directed path between every pair of nodes. 

In this work, we consider signed digraphs that may contain multiple leader nodes. A leader node is defined as either a root node or a node that is part of a \textit{rooted strongly-connected component (rooted SCC)}. A rooted SCC is a strongly connected subgraph without incoming edges. A \textit{leader group} is either a single root node (representing a single-node leader group) or an entire rooted SCC (representing multiple leader nodes interconnected in a strongly connected subgraph). If the graph contains at least one leader group, the remaining nodes are referred to as followers.--- See Figure \ref{leadergroups}.

A signed graph is said to be \textit{structurally balanced} (SB) if it can be split into two disjoint sets of vertices $\mathcal{V}_{1}$ and $\mathcal{V}_{2}$, where $\mathcal{V}_{1} \cup \mathcal{V}_{2} = \mathcal{V}, \mathcal{V}_{1} \cap \mathcal{V}_{2} = \emptyset$ such that for every $\nu_{i}, \nu_{j} \in \mathcal{V}_{p}, p \in \{ 1,\, 2 \}$, if $a_{ij} \geq 0$, while for every $\nu_{i} \in \mathcal{V}_{p}, \nu_{j} \in \mathcal{V}_{q}$, with $p,q \in \{ 1, \, 2 \}, p \neq q$, if $a_{ij} \leq 0$. It is {\it structurally unbalanced} (SUB), otherwise \cite{altafini2012_6329411}. The signed Laplacian matrix $L_{s}=[{\ell _{s_{ij}}}] \in \mathbb{R}^{N \times N}$ associated with $\mathcal{G}_s$ is defined as ${\ell _{s_{ij}}} := {\sum\limits_{k \leq {N}} {{ \lvert a_{ik} \rvert }} }$, if ${i = j}$; and ${\ell _{s_{ij}}} := { - {a_{ij}}}$, if ${i \ne j}$ \cite{altafini2012_6329411,Shi2019Dynamics}.
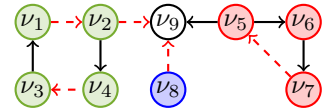
\begin{figure}[!t]\centering
\begin{tikzpicture}[node distance={9mm}, thick,main/.style = {draw, circle}] 
    \definecolor{pastelgreen}{rgb}{0.4660 0.6740 0.1880}
	\tikzset{mynode/.style={circle,draw,minimum size=13pt,inner sep=0pt,thick},}
    % \draw[pastelgreen,very thin, fill=pastelgreen!5](-0.5,0.5) rectangle(1.25,-1.5);
    % \node[align=center,pastelgreen] at (0.375,-1.9) {SB-rooted\\SCC};
    % \draw[red,very thin, fill=red!5](2.35,0.5) rectangle(4,-1.5);
    % \node[align=center,red] at (3.175,-1.9) {SUB-rooted\\SCC};
    % \draw[blue,very thin, fill=blue!5](1.4,-0.5) rectangle(2.2,-1.5);
    % \node[align=center,blue] at (1.8,-1.9) {root\\node};
	\node[mynode, fill=pastelgreen, fill opacity=0.2, draw=pastelgreen, text opacity=1] (1) {$\nu_1$}; 
	\node[mynode, fill=pastelgreen, fill opacity=0.2, draw=pastelgreen, text opacity=1] (2) [right  of=1] {$\nu_2$};
	\node[mynode, fill=pastelgreen, fill opacity=0.2, draw=pastelgreen, text opacity=1] (3) [below of=1] {$\nu_3$};
	\node[mynode, fill=pastelgreen, fill opacity=0.2, draw=pastelgreen, text opacity=1] (4) [below of=2] {$\nu_4$};
    \node[mynode] (9) [right of=2] {$\nu_9$};
    \node[mynode, fill=red, fill opacity=0.2, draw=red, text opacity=1] (5) [right of=9]{$\nu_5$}; 
	\node[mynode, fill=red, fill opacity=0.2, draw=red, text opacity=1] (6) [right of=5] {$\nu_6$};
	\node[mynode, fill=red, fill opacity=0.2, draw=red, text opacity=1] (7) [below of=6] {$\nu_7$};
    \node[mynode, fill=blue, fill opacity=0.2, draw=blue, text opacity=1] (8) [below of=9] {$\nu_8$};
    \draw[red, dash pattern=on 1mm off 1mm][->](1) -- node[midway, above] {}(2);
	\draw[<-] (1) -- node[midway, left] {}(3);
	\draw[->] (2) -- node[midway, right] {}(4);
	\draw[<-, color=red, dash pattern=on 1mm off 1mm] (3) -- node[midway, below] {}(4);
    \draw[->](5) -- node[midway, above] {}(6);
    \draw[->](6) -- node[midway, above] {}(7);
    \draw[red, dash pattern=on 1mm off 1mm][->](7) -- node[midway, above] {}(5);
    \draw[red, dash pattern=on 1mm off 1mm][->](2) -- node[midway, above] {}(9);
    \draw[->](5) -- node[midway, above] {}(9);
    \draw[red, dash pattern=on 1mm off 1mm][->](8) -- node[midway, above] {}(9);
	\end{tikzpicture}
\caption{A signed digon sign-symmetric digraph containing $3$ leader groups, where the black edges represent cooperative interactions and the dashed red edges represent antagonistic interactions. The first leader group is a SB-rooted SCC containing the leader nodes $\nu_1, \nu_2, \nu_3,$ and $\nu_4$, the second leader group is a SUB-rooted SCC containing the leader nodes $\nu_5, \nu_6,$ and $\nu_7$, and the third leader group is a root (leader) node, $\nu_8$. The node $\nu_9$ is a follower node.} \label{leadergroups}
\end{figure}

We now recall definitions of the signed incidence matrices of a signed digraph. Consider a signed graph $\mathcal{G}_s$ containing $N$ nodes and $M$ edges. The signed
incidence matrix $E_s \in \mathbb{R}^{N \times M}$ of $\mathcal{G}_s$ is defined as
    \begin{equation}\label{434}
	{[E_{s}]_{ik}} :=  \left\{ \begin{matrix*}[l]
		{+1,}&{\text{if}\ \varepsilon_k = (\nu_i, \nu_j);}\\
		{-1,}&{\text{if } \nu_i,\nu_j\ \text{are cooperative and } \varepsilon_k = (\nu_j, \nu_i);}\\ %terminal node of the edge}\ e_k;}\\
{+1,}&{\text{if } \nu_i,\nu_j\ \text{are competitive and }\varepsilon_k = (\nu_j, \nu_i);}\\
%{}&{\nu_i \text{ is the terminal node of}\ \varepsilon_k ;}\\
{0,}&{\text{otherwise}},
\end{matrix*}\right.
\end{equation}
and the signed in-incidence matrix $E_{{s \odot}} \in \mathbb{R}^{N \times M}$ of $\mathcal{G}_s$ is defined as
\begin{equation} \label{490}
{[E_{s \odot}]_{ik}} := \left\{ \begin{matrix*}[l]
{-1,}&{\text{if } \nu_i,\nu_j\ \text{are cooperative and } \varepsilon_k = (\nu_j, \nu_i);}\\
%{}&{\nu_i \text{ is the terminal node of}\ \varepsilon_k ;}\\
{+1,}&{\text{if } \nu_i,\nu_j\ \text{are competitive and } \varepsilon_k = (\nu_j, \nu_i);}\\
%{}&{\nu_i \text{ is the terminal node of}\ \varepsilon_k ;}\\
{0,}&{\text{otherwise}},
\end{matrix*}\right. 
\end{equation}
where $\varepsilon_k$ is the oriented edge interconnecting nodes $\nu_i$ and $\nu_j$,\ $k \leq M,\ i,j \leq N$. The signed Laplacian $L_{s} \in \mathbb{R}^{N \times N}$ and the signed edge Laplacian $L_{e_{s}} \in \mathbb{R}^{M \times M}$ of a signed digraph, they are given as
\begin{align}\label{498}
L_{s} = E_{s \odot}E_{s}^\top,\quad L_{{e_s}} = E_{s}^\top E_{{s \odot}}.
\end{align}
The signed Laplacian of a signed digraph is not symmetric, and its eigenvalues all have nonnegative real parts.

% The signed Laplacian matrix $L_{s}=[{\ell _{s_{ij}}}] \in \mathbb{R}^{N \times N}$ associated with $\mathcal{G}_s$ is defined as \cite{altafini2012_6329411,Shi2019Dynamics} 
% \begin{equation}\label{Ls}
% 	{\ell _{s_{ij}}} = \left\{ {\begin{array}{*{20}{c}}
% 			{\sum\limits_{k \leq {N}} {{ \lvert a_{ik} \rvert }} }&{i = j}\\
% 			{ - {a_{ij}}}&{i \ne j}.
% 	\end{array}} \right.
% \end{equation}

\section{Model and Problem Formulation}\label{section2}
Consider a group of $N$ first-order agents interconnected over a signed digraph $\mathcal{G}_{s}$ with $M$ cooperative and antagonistic edges. The dynamics of each agent are given by
\begin{align}\label{FO}
	\dot x_i = u_i,\quad i \in \{1, 2, \dots, N \},
\end{align}
where $x_i \in \mathbb{R}$ is the state of the $i$th agent, and $u_i \in \mathbb{R}$ is the control input. %For notational simplicity and without loss of generality, we assume that $x_i \in \mathbb{R}$, but all contents of this paper apply to systems of higher dimension $x_i \in \mathbb{R}^n, n \geq 1$, using a Kronecker product. The agents interact on a signed digraph $\mathcal{G}_{s}$ that contains $N$ nodes and $M$ edges. 
The system \eqref{FO} is interconnected with the control law
\begin{align}\label{CL}
	u_i = -k_{1} \sum_{j=1}^{N} \lvert a_{{ij}} \rvert \left[ x_i - \mbox{sign}(a_{{ij}} )x_j \right], 
\end{align}
%or in vector form
%\begin{align}\label{CL_vect}
%	u_{\phi} = -k_1 L_{s_{\phi}}^\top x_{\phi},
%\end{align}
where $k_{1}>0$, and $A = [a_{{ij}}]$ is the adjacency matrix, with $a_{{ij}} \in \{ 0, \pm 1 \}$ representing the adjacency weight between nodes $\nu_j$ and $\nu_i$. It is well known that under the distributed control law \eqref{CL}, agents interconnected over a SB signed digraph achieve bipartite consensus if and only if the signed graph contains a directed spanning tree, while those over a SUB digraph reach trivial consensus provided that the graph contains a directed spanning tree and has no root node \cite{altafini2012_6329411,hu2014emergent}. Thus, we pose the following on the connectivity of the signed graph.
\begin{ass}\label{ass1}
    The signed digraph contains a directed spanning tree.
\end{ass}

In this article, we analyze the behavior of the multi-agent system \eqref{FO} in closed loop with the control law \eqref{CL} and under the assumption that the edges interconnecting the agents are cooperative and antagonistic. The possible achievable control objectives for the system \eqref{FO}-\eqref{CL} interconnected on a signed digraph depend on the \textit{structural balance} property and the presence of one or multiple leaders. In particular, the presence of rooted SCCs and root nodes plays a fundamental role in shaping the convergence behavior of the agents. A rooted SCC affects the asymptotic behavior of the network, depending on whether it is SB or SUB. When the digraph contains a single root node, the system achieves a leader-following bipartite consensus, where all agents converge either to the leader's state or its opposite state in the SB case. In such cases, the root node or the rooted SCC acts as a \textit{leader group}, while the remaining agents behave as \textit{followers}. Formally, a leader group is defined as either a single root node or a rooted SCC composed of multiple nodes. By contrast, when the digraph contains multiple root nodes or rooted SCCs, each leader group independently influences the rest of the network. Rather than converging to a single consensus value or bipartite consensus configuration, the agents achieve bipartite containment, that is, their states converge to a region determined by the states of the leader groups. In this particular case, the digraph does not admit a directed spanning tree, since the distinct leader groups are not mutually reachable. Therefore, for signed digraphs with multiple leaders groups, we pose the following connectivity assumption.
\begin{ass}\label{ass2} The signed digraph contains $m$ leader groups, where $l_1$ is the number of root nodes, $l_{2_{SB}}$ is the number of SB-rooted SCCs, and $l_{2_{SUB}}$ is the number of SUB-rooted SCCs with $m=l_1 + l_2 > 1$ and $l_2 = l_{2_{SB}} +l_{2_{SUB}}$; $\mathcal L$ and $\mathcal F$ are the sets containing the index corresponding to the nodes in the leader groups and followers respectively. 
\begin{enumerate}
	\item The network contains $k$ leader nodes, which can be organized into $m$ leader groups of $p_i$ nodes included in a strongly connected subgraph (or $p_i=1$ if it is a single root node), where $1 < m \leq k < N,\ i \leq m$, and $\sum_{i=1}^m p_i = k$.
	\item There exists at least a path from $\nu_i$ to $\nu_j$, $\forall j \in \mathcal{F}$ and $i \in \mathcal{L}$.
    %\item Given each follower $\nu_j$, with $j \in \mathcal{F}$, there exists at least one leader $\nu_i$, with $i \in \mathcal{L}$, such that there exists at least one path from $\nu_i$ to $\nu_j$.
\end{enumerate}
\end{ass}
\begin{remark}
    Assumption \ref{ass2} is derived from the definition of leader proposed in \cite{meng2017bipartite}. If we consider the signed digraph in Figure \ref{leadergroups}, the leader set is given as $\mathcal L = \{ \nu_1, \nu_2, \nu_3, \nu_4, \nu_5, \nu_6, \nu_7, \nu_8 \}$, the follower set is given as $\mathcal F = \{ \nu_9 \}$, and $l_1 = l_{2_{SB}} = l_{2_{SUB}} = 1$. Regarding Item 1 of Assumption \ref{ass2}, we have $k=8$ leader nodes organized into $m=3$ leader groups. The first leader group contains $p_1 = 4$ nodes, the second leader group contains $p_2 = 3$ nodes, and the third leader group contains $p_3 = 1$ node, where $\sum_{i=1}^m p_i = p_1 + p_2 + p_3 = 8.$
\end{remark}

The achievable control objective for \eqref{FO} in closed loop with the distributed control law \eqref{CL} and interconnected over a SB digraph:
\begin{itemize}
    \item containing a directed spanning tree (Assumption \ref{ass1} holds) is to ensure agents achieve \textit{bipartite consensus},  where agents converge to the same value in modulus but not in sign, that is,
\begin{align}\label{obj_BC}
	\lim_{t \to \infty} \left[ x_{i}(t) - \mbox{sign}(a_{{ij}})x_{j}(t) \right] = 0, \ \forall i,j \leq N.
\end{align}
    \item containing $m$ leader groups, under Assumption \ref{ass2}, is to ensure agents achieve \textit{bipartite containment}, that is,
		\begin{equation}\label{eq:containment_SB} 
		  \lim_{t \to \infty} [ x_j(t) - \max_{i \in \mathcal{L}} (s_{i} x_{i}(t)) ] [ x_j(t) - \min_{i \in \mathcal{L}} (s_{i} x_{i}(t))] \leq 0, 
	\end{equation}
for each $j \in \mathcal{F}$, where $s_{i} =1$ if agent $\nu_j$ is cooperative with leader $\nu_i$, and $s_{i} = -1$ otherwise.
\end{itemize} 

The achievable control objective for \eqref{FO} in closed loop with the distributed control law \eqref{CL} interconnected over a SUB digraph:
\begin{itemize}
    \item containing a directed spanning tree (Assumption \ref{ass1}) and a SUB-rooted SCC, without a root (leader) node, is to ensure agents achieve \textit{trivial consensus}, where all agents converge to zero, that is,
    \begin{align}\label{obj_C}
	   \lim_{t \to \infty} x_{i}(t) = 0, \quad \forall i \leq N_{\phi}.
    \end{align} 
    \item containing a directed spanning tree (Assumption \ref{ass1}) and SB-rooted SCC or a root node, is to ensure agents achieve \textit{interval bipartite consensus}, that is,
    \begin{align}\label{obj_IntBC}
	   \lim_{t \to \infty} x_{i}(t) \in [-\theta, \theta], \quad \forall i \leq N,
    \end{align} 
    where $\theta > 0$.
    \item containing $m$ leader groups, under Assumption \ref{ass2}, is to ensure agents achieve \textit{bipartite containment}, that is, 
    \begin{equation}  \label{eq:containment_SUB} 
	   \lim_{t \to \infty} \big[\, \lvert x_{j}(t) \rvert - \max_{i \in \mathcal{L}}\lvert x_{i}(t) \rvert \,\big] \leq 0,\quad j \in \mathcal{F}.
\end{equation}
\end{itemize}
\begin{remark}
If the considered digraph is SUB and contains multiple leader groups, where all rooted SCCs are SUB and the digraph contains a root node, the bipartite containment objective \eqref{eq:containment_SUB} reduces to interval bipartite consensus \eqref{obj_IntBC}. If the digraph contains only SUB-rooted SCCs, the bipartite containment objective \eqref{eq:containment_SUB} reduces to trivial consensus \eqref{obj_C}.
\end{remark}

% \begin{figure}[!t]\centering
% \begin{tikzpicture}[node distance={9mm}, thick,main/.style = {draw, circle}] 
% \definecolor{pastelgreen}{rgb}{0.4660 0.6740 0.1880}
% \tikzset{mynode/.style={circle,draw,minimum size=15pt,inner sep=0pt,thick},}

% \node[mynode] (1) {$\nu_1$}; 
% \node[mynode] (2) [right  of=1] {$\nu_2$};
% \node[mynode] (3) [below of=1] {$\nu_3$};
% \node[mynode] (4) [below of=2] {$\nu_4$};

% \node[mynode] (5) [right of=9]{$\nu_5$}; 
% \node[mynode] (6) [right of=5] {$\nu_6$};
% \node[mynode] (7) [below of=6] {$\nu_7$};

% \node[mynode,blue] (8) [below of=4] {$\nu_8$};
% \node[mynode,blue] (9) [below of=7] {$\nu_9$};
% \node[mynode,blue] (10) [left of=9] {$\nu_{10}$};

% \foreach \from/\to in {1/2,2/4,4/3,3/1, %
% 5/6,6/7,7/5, %
% 5/10,7/10, 5/9, 4/8,3/8}{
% \draw[<-] (\from) -- (\to);}

% \foreach \from/\to in {2/5,5/4,5/8,4/10}{
% \draw[->,red, dash pattern=on 1mm off 1mm] (\from) -- (\to);}

% \draw[pastelgreen,fill=pastelgreen!5] (0,-3.05) rectangle(2.8,-2.95);
% \draw[pastelgreen,fill=pastelgreen!5] (6.2,-3.05) rectangle(7.1,-2.95);

% \draw[->,thick] (0,-3) -- (5,-3);

% \node[align=center,pastelgreen] at (0.375,-1.9) {converge zones for follower $\nu_1$};

% \end{tikzpicture}
% \caption{
% For each follower $i\in \mathcal{F}$, the achievable convergence zone is determined by a ``containment set'', which is determined by convex hull of the leaders' states, taken with the sign that captures the type of interaction between the follower and leader, positive if cooperative, negative if antagonistic. } 
% \end{figure}

\begin{remark}
The bipartite containment objective of \eqref{eq:containment_SB}, \eqref{eq:containment_SUB} can be interpreted as follows. 
When the digraph is SB \eqref{eq:containment_SB}, for each follower $i\in \mathcal{F}$, the achievable convergence zone is determined by a containment set defined by the absolute value of the leaders' states, taken with the sign that captures the type of interaction between the follower and leader, positive if cooperative, negative if antagonistic. In particular, this leads to two symmetric convergence zones, determined by the gauge permutation capturing the SB property of the signed digraph (Lemma~\ref{lemma_gauge} in Section \ref{section3a1}). Instead, when the digraph is SUB \eqref{eq:containment_SUB}, the containment set is defined only by the absolute value of the leaders' states, leading to a unique achievable convergence zone.
For more details, simulation examples, and illustration of convergence zones, see \cite[Section V]{sekercioglu2023exponential}.
% For each mode, assuming that a signed digraph is SB means that each follower $j \in \mathcal{F}$ will converge to one of two symmetric zones defined by the leaders' states. In particular, the boundaries of these zones are given in \eqref{eq:containment_SB} by the max/min values across the states of the leaders, after 
% The bipartite containment objective in \eqref{eq:containment_SUB} is a more general version of the  
% the bipartite containment objective in \eqref{eq:containment_SB} is a refined version of \eqref for the more restrictive case of SB digraphs with multi-leader groups. 
% For an SB signed digraph, the presence of two antagonistic disjoint subsets leads to two symmetric convergence zones, as described by \eqref{eq:containment_SB}, defined by the cooperative and symmetric antagonistic leader states. In contrast, for a SUB signed digraph, where no opposing subsets exist, agents converge to a single zone, as described by \eqref{eq:containment_SUB}, determined by all leaders' final states and symmetric final states. For more details and simulation examples, see Section V in \cite{sekercioglu2023exponential}.
\end{remark}

% \begin{remark}
%     The bipartite containment objective in \eqref{eq:containment_SB} is a refined version of the bipartite containment objective in \eqref{eq:containment_SUB} for the more restrictive case of SB digraphs with multi-leader groups. For a SB signed digraph, the presence of two antagonistic disjoint subsets leads to two symmetric convergence zones, as described by \eqref{eq:containment_SB}, defined by the cooperative and symmetric antagonistic leader states. In contrast, for a SUB signed digraph, where no opposing subsets exist, agents converge to a single zone, as described by \eqref{eq:containment_SUB}, determined by all leaders' final states and symmetric final states. For more details and simulation examples, see Section V in \cite{sekercioglu2023exponential}.
% \end{remark}

In this article, we show that, under Assumptions~\ref{standing_ass}--\ref{ass2}, depending on whether the associated signed digraph is SB or SUB and contains multiple leader groups or not, synchronization in the signed OMAS leads to bipartite consensus \eqref{obj_BC}, trivial consensus \eqref{obj_C}, interval bipartite consensus \eqref{obj_IntBC} or bipartite containment \eqref{eq:containment_SB}, \eqref{eq:containment_SUB}. In the case agents achieve trivial consensus or bipartite consensus, their edges states, defined as
\begin{equation}\label{def_e}
    e_k = x_{i} - \mbox{sign}(a_{{ij}})x_{j},\quad \varepsilon_k =(\nu_j,\nu_i) \in \mathcal{E},
\end{equation}
where $k \leq M$ denotes the index of the interconnection between the $j$th and $i$th agents, converge to zero. However, as the networks contain a priori rooted SCCs or multiple root nodes, the resulting Laplacians can also have multiple zero eigenvalues in some cases. This also results, in general, in multiple convergence points for agents and their edge states, which means that edge states do not always converge to zero. Then, following the framework laid in \cite{DYNCON-TAC16} and extending it to signed networks with associated Laplacians containing multiple zero eigenvalues and to the signed edge-based formulation, we define the weighted average system for the edge states.

Let $\xi$ be the number of zero eigenvalues of $L_{e_s}$. Then, we define the weighted edge average state as
\begin{equation}\label{e_m}
	e_{m} := \sum_{i=1}^{\xi} v_{r_{i}}v_{l_{i}}^{\top} e,
\end{equation}
where $v_{r_{i}}$ and $v_{l_{i}}$ are the right and left eigenvectors associated with the zero eigenvalues of the edge Laplacian, $e := [e_1 \ e_2 \ \cdots \ e_{M}]^\top$, and $e_k$ is defined in \eqref{def_e}. We define the synchronization errors as
\begin{align}\label{def_e_bar}
	\bar e &= e - e_{m} = [I - \sum_{i=1}^{\xi} v_{r_{i}}v_{l_{i}}^{\top}]e,
\end{align}
where $\bar{e} := [\bar e_1 \ \bar e_2 \ \dots \ \bar e_{M}]^\top$. Then, the control objective is equivalent to making the synchronization errors converge to zero, that is, 
\begin{align}\label{obj}
	\lim_{t \to \infty} \bar{e}_{k}(t) = 0, \quad \forall k \leq M.
\end{align}
% \begin{remark}
% In the case where the nullspaces of $L_{e_s}^\top$ and $E_s$ are equal to each other, $e_{m} = 0$, and the synchronization errors are equal to the edge states.
% \end{remark}

Before introducing our main results in Section~\ref{section5}, we remind in Section~\ref{section4} some properties of the signed edge Laplacians from \cite[Section IV]{TACPelin} that are useful for establishing our main results.%, using notions on signed graphs and the signed edge-based formulation from \cite{altafini2012_6329411} and \cite{du2019further}, respectively. %We extend the method in \cite{ECC20_EP} and \cite{csekerciouglu2024distributed}, to construct strict Lyapunov functions for linear systems with multiple null eigenvalues and edge Laplacians.

\section{Signed edge-based formulation}\label{section3a1}
\subsection{Spectral properties of signed edge Laplacians}
For a general digraph, the following properties hold regarding similarity transformations, the positive (semi-)definiteness of its Laplacians, and the dimension of their null space \cite{altafini2012_6329411, du2018edge, du2019further}. In particular, Lemmata \ref{lemmaE_in}, \ref{lemma3a}--\ref{lemma_multipleleaders} present original contributions. We first start by the rank properties of the in-incidence matrix and then consider SB and SUB digraphs containing a directed spanning tree, respectively. Similar properties for undirected signed graphs can be found in \cite{CDC2025OMAS}.

\begin{lemma}\label{lemmaE_in}
For a signed digraph containing a directed spanning tree, $\mbox{rank}(E_{s \odot}) = N-1$ if the digraph contains a root node, and $\mbox{rank}(E_{s \odot}) = N$, otherwise.
\end{lemma}
\begin{proof} A directed spanning tree of $N$ agents contains exactly $N - 1$ independent edges. If the digraph contains a root node, this root node has no incoming edges, which implies that the corresponding row in the in-incidence matrix is filled by zeros. In contrast, all other nodes have at least one incoming edge, resulting in at least one nonzero entry (either $1$ or $-1$) in their respective rows of the in-incidence matrix. Consequently, $\mbox{rank}(E_{s \odot}) = N-1$. Conversely, in the case the digraph does not contain a root node, we have $\mbox{rank}(E_{s \odot}) = N$. \end{proof}

We can transform the Laplacian and incidence matrices of SB signed graphs into matrices corresponding to unsigned graphs using some similarity (gauge) transformations. In the case the graph is SUB, these transformations cannot be used.
\begin{lemma}(Gauge transformation \cite{altafini2012_6329411})\label{lemma_gauge}
For a SB signed graph, there exists a diagonal matrix $D \in \mathfrak{D}$, where $\mathfrak{D} = \left\{ D = \mbox{diag}(d) \mid d = [d_1\ d_2\ \cdots \ d_N],\ d_i \in \{ -1, 1\},\ i \leq N \right\},$ such that all off-diagonal elements of \( D L_s D \) are non-positive, \textit{i.e.,} \( D L_s D \) is an M-matrix.%has the structure of a Laplacian matrix corresponding to an unsigned graph.
\end{lemma}

\begin{lemma}(Edge-gauge transformation \cite[Lemma 4]{du2018edge})\label{lemma_edgegauge}
For a SB signed graph, there exist diagonal matrices \( D = \operatorname{diag}(d) \), from Lemma \ref{lemma_gauge}, and \( D_e = \operatorname{diag}(d_e) \), where $d_e = [d_{e_1}\ \cdots\ d_{e_M}], \quad i \leq M,$ with $d_{e_i} = 1$ if $\nu_i \in \mathcal{V}_1,$ and $d_{e_i} = -1$ if $\nu_i \in \mathcal{V}_2,$ where $\nu_i$ is the initial node of the edge, such that $E = D E_s D_e$ has the structure of an incidence matrix corresponding to an unsigned graph.
\end{lemma}

%\subsubsection{Rank properties of the  signed Laplacian and signed edge Laplacian}
We now analyze the rank properties of the signed Laplacian and the signed edge Laplacian when the digraph contains a directed spanning tree, considering both the SB (Lemma~\ref{lemma3a}) and SUB (Lemma~\ref{lemma3b}) cases. Then, Lemma~\ref{NE_SUB} studies the null space of these matrices. Finally, Lemma~\ref{lemma_multipleleaders} addresses the case where the digraph lacks a spanning tree and multiple leader groups are present.

For SB signed digraphs, we have the following regarding the rank properties of the Laplacian and incidence matrices.
\begin{lemma}\label{lemma3a}
	For a SB digraph containing a directed spanning tree, the following holds. 
	\begin{enumerate}[{\normalfont (i)}]
		\item $L_{s}$ has a simple zero eigenvalue and its other eigenvalues have positive real parts.
        \item $\mbox{rank}(E_s) =N-1$.
		\item $\mbox{rank}(L_{s}) = \mbox{rank}(L_{e_{s}}) = N-1$.
	\end{enumerate}
\end{lemma}

\begin{proof}
\textbf{(i)} By assumption, the digraph is SB and contains a directed spanning tree, so we apply the gauge transformation (Lemma \ref{lemma_gauge}) to $L_{s}$, and the statement follows directly from \cite[Lemma 3.3]{1431045}.\\
\textbf{(ii)} By assumption, the digraph is SB and contains a directed spanning tree, so we apply the edge-gauge transformation (Lemma \ref{lemma_edgegauge}) to $E_{s}$. The statement follows directly from \cite{thulasiraman2011graphs}.\\
\textbf{(iii)} We follow the steps in the proof of \cite[Lemma 2]{zeng2014nonlinear}. Suppose $\lambda \neq 0$ is an eigenvalue of $L_{s}$, which is associated with a nonzero eigenvector $v_{r}$, such that $L_{s} v_{r} = E_{s \odot}E_{s}^\top v_{r} = \lambda v_{r} \neq 0.$ It is clear that $E_{s}^\top v_{r} \neq 0$. Let $\bar v_{r} = E_{s}^\top v_{r}$. Then, by left-multiplying $E_{s}^\top$ on both sides, we obtain $E_{s}^\top E_{s \odot}E_{s}^\top v_{r} = E_{s}^\top \lambda v_{r}$. By replacing \eqref{498} in the latter, we obtain $L_{e_{s}} \bar v_{r} = \lambda \bar v_{r},$ which means $\lambda$ is also an eigenvalue of $L_{e_{s}}$. Consequently, the nonzero eigenvalues of $L_{s}$ and $L_{e_{s}}$ are equal to each other. Then, according to Item (ii), proven above, and from Lemma \ref{lemmaE_in}, we have $\mbox{rank}(L_{{s}}) \leq \min \{\mbox{rank}(E_{s \odot}), \mbox{rank}(E_{s}^\top)\} = N-1$ and $\mbox{rank}(L_{e_{s}}) \leq \min \{\mbox{rank}(E_{s}^\top), \mbox{rank}(E_{s \odot})\} = N-1$. Since they both have $N-1$ nonzero eigenvalues, the statement follows.
\end{proof}

Unlike the SB case, when the digraph is SUB and contains a directed spanning tree, then its structure falls into one of three distinct cases: the presence of a SUB-rooted SCC, a SB-rooted SCC, or a root node, as detailed in the following lemma.
\begin{lemma}\label{lemma3b}
For a SUB digraph containing a directed spanning tree, the following holds. 
\begin{enumerate}[{\normalfont (i)}]
    \item $L_{s}$ has only eigenvalues with positive real parts if and only if the digraph contains a SUB-rooted SCC.
    \item $L_{s}$ has a simple zero eigenvalue and its other eigenvalues have positive real parts if and only if the digraph contains a SB-rooted SCC or a root node.
    \label{lemma3b_ii}
    \item $\mbox{rank}(E_s) = N$.
     \item $\mbox{rank}(L_{s}) = \mbox{rank}(L_{e_{s}})$ if and only if the digraph contains a SUB-rooted SCC or a root node. Moreover, 
     \[
     \mbox{rank}(L_{s}) = \begin{cases}
         N, & \text{if $\mathcal{G}_s$ contains SUB-rooted SCC}
         \\
         N-1,  & \text{if $\mathcal{G}_s$ contains a root node.}
     \end{cases}
     \]\label{lemma3b_iii}
     \item $N-1=\mbox{rank}(L_{s})<\mbox{rank}(L_{e_{s}})=N$ if and only if the digraph contains a SB-rooted SCC.
\end{enumerate}
\end{lemma}
\begin{proof}
\textbf{(i)}-\textbf{(ii)} follow from \cite[Lemma 5]{du2019further}.
\\
% ----------- (iii)
\textbf{(iii)} Since, SB and SUB are mutually exclusive properties, the statement follows (\cite[Lemma 2]{du2019further}).
\\
% ----------- (iv)
\textbf{(iv), ``$\Leftarrow$'': }
Assume that the digraph contains a SUB-rooted SCC or a root node. We separate the two cases. 

When a SUB digraph contains a SUB-rooted SCC, from Item (i), we know that $L_s$ has $N$ nonzero eigenvalues. On the other hand, from Lemma \ref{lemmaE_in}, $\mbox{rank}(E_{s \odot}) = N$ since the digraph does not contain a root node. So, $\mbox{rank}(L_{{s}}) \leq \min \{\mbox{rank}(E_{s \odot}), \mbox{rank}(E_{s}^\top)\} = N$ and $\mbox{rank}(L_{e_{s}}) \leq \min \{\mbox{rank}(E_{s}^\top), \mbox{rank}(E_{s \odot})\} = N$. Since, from the proof of Item (iii) of Lemma \ref{lemma3a}, $L_{s}$ and $L_{e_{s}}$ have the same nonzero eigenvalues, meaning that they both have $N$ nonzero eigenvalues, the statement $\mbox{rank}(L_{s}) = \mbox{rank}(L_{e_{s}}) = N$ follows.

Instead, when a SUB digraph contains a root node, from Item (ii), we know that $L_s$ has a unique zero eigenvalue and $N-1$ nonzero eigenvalues. Moreover, from Lemma \ref{lemmaE_in}, $\mbox{rank}(E_{s \odot}) = N-1$ since the digraph contains a root node. Then, the statement $\mbox{rank}(L_{s}) = \mbox{rank}(L_{e_{s}}) = N-1$ follows.
\\
\textbf{(iv), ``$\Rightarrow$'':}
Assume that $\mbox{rank}(L_{s}) = \mbox{rank}(L_{e_{s}})$. Moreover, assume (by contradiction) that the digraph contains an SB-rooted SCC (i.e., it does not contain a SUB-rooted SCC or a root node). 
Using Item (ii), this implies $\mbox{rank}(L_{s})<N$.
Moreover, from \cite[Section~0.4.5]{Horn}, we know that 
$\mbox{rank}(L_{e_{s}}) \ge \mbox{rank}(E_{s}^\top) + \mbox{rank}(E_{s \odot}) - N $.
Then, using Item (iii) and Lemma~\ref{lemmaE_in}, we conclude that $\mbox{rank}(L_{e_{s}}) = N \ne N-1 = \mbox{rank}(L_{s}) $, obtaining a contradiction.
\\
\textbf{(v)} follows from the proof of Item (iv).
% \textbf{(v)} Observe that Item~\eqref{lemma3b_iii} of Lemma~\ref{lemma3b} implies that $\mbox{rank}(L_{s})\ne \mbox{rank}(L_{e_{s}})$ if and only if the digraph contains an SB-rooted SCC. 
% Then, the statement follows from $\mbox{rank}(L_{s}) \le \mbox{rank}(L_{e_{s}})$ and Item~\eqref{lemma3b_ii} of Lemma~\ref{lemma3b}, from which we conclude that $N-1=\mbox{rank}(L_{s})<\mbox{rank}(L_{e_{s}})=N$ if and only if the digraph contains an SB-rooted SCC.
%However, since $\mbox{rank}(L_{s}) \le \mbox{rank}(L_{e_{s}})$, then, using Item~\eqref{lemma3b_ii} of Lemma~\ref{lemma3b}, we can conclude that $N-1=\mbox{rank}(L_{s})<\mbox{rank}(L_{e_{s}})=N$ if and only if the digraph contains an SB-rooted SCC.
\end{proof}

\begin{lemma} (\cite[Lemma 7]{du2019further}) \label{NE_SUB}
For a SB digraph containing a directed spanning tree, $\mathcal{N}(L_{e_s}^\top) = \mathcal{N}(E_s)$ holds. For a SUB digraph containing a directed spanning tree:
    	\begin{enumerate}[{\normalfont (i)}]
		\item if the signed digraph does not contain a root node, then $\mathcal{N}(L_{e_s}^\top) = \mathcal{N}(E_s)$ holds.
		\item if the signed digraph contains a root node, then $\mathcal{N}(L_{e_s}^\top) \neq \mathcal{N}(E_s)$ holds.
	\end{enumerate}
\end{lemma}
\begin{proof} The statement follows directly from the definition of $L_{e_s}$ in \eqref{498}, the rank properties of $E_s$ and $L_{e_s}$ in Lemmata \ref{lemma3a} and \ref{lemma3b}, and the rank properties of $E_{s \odot}$ of in Lemma \ref{lemmaE_in}.
\end{proof}

Now, we consider signed digraphs containing multiple leader groups, meaning that the signed digraph does not contain a directed spanning tree, since root nodes and rooted SCCs have no incoming edges--- See Figure \ref{leadergroups}. The following lemma is an original contribution of this paper.

\begin{lemma}\label{lemma_multipleleaders}
Consider a signed digraph containing $m$ leader groups, where $m > 1$. Let $l_1\geq0$ be the number of root nodes, $l_{2_{\mathrm{SB}}}\geq0$ be the number of SB-rooted SCCs, and $l_{2_{\mathrm{SUB}}}\geq0$ be the number of SUB-rooted SCCs, where $m=l_1 + l_2$ and $l_2 = l_{2_{\mathrm{SB}}} +l_{2_{\mathrm{SUB}}}$. Assume that given each follower $\nu_j$, there exists at least one leader $\nu_i$, such that there exists at least one path from $\nu_i$ to $\nu_j$. Then, the following holds.
\begin{enumerate}[{\normalfont (i)}]
    \item $L_s$ has $l_1 + l_{2_{\mathrm{SB}}}$ zero eigenvalues and its other eigenvalues have positive real parts.
    \item $\mbox{rank}(E_{s \odot}) = N-l_1$.% if all the rooted SCCs are SB, and $\mbox{rank}(E_{s \odot}) = N-m$.
    \item $\mbox{rank}(E_{s}) = N-1$ if the digraph is SB, and $\mbox{rank}(E_{s}) = N$ otherwise.
    \item $\mbox{rank}(L_{s}) = N-l_1-l_{2_{\mathrm{SB}}}$.
    \item $N - l_1 - 1 \leq \mbox{rank}(L_{e_{s}}) \leq N-l_1$. Moreover, if the digraph is SB, 
        \begin{enumerate}[{\normalfont (a)}]
            \item $\mbox{rank}(L_{e_{s}})=N-1$, if $l_1=0$,% and $l_{2_{SB}}\geq 2$,
            \item $\mbox{rank}(L_{e_{s}})=N-l_1$, if $l_{2_{\mathrm{SB}}}=0$.% and $l_1 \geq 2$,
        \end{enumerate}
        Otherwise, $\mbox{rank}(L_{e_{s}}) = N-l_1$.% for $l_1 \geq 0$.
    \item $\mathcal{N}(L_{e_s}^\top) = \mathcal{N}(E_s)$, if all leader groups are rooted SCCs and $\mathcal{N}(L_{e_s}^\top) \neq \mathcal{N}(E_s)$, if the digraph contains at least one root node. %or if the digraph contains a root node and is SB
	\end{enumerate}
\end{lemma}
\begin{proof} \textbf{(i)} From \cite{altafini2012_6329411}, we have that a strongly-connected component adds a zero eigenvalue to the signed Laplacian if and only if it is SB. This, together with \cite[Lemma 1]{csekerciouglu2024distributed} proves that $L_s$ has $l_1 + l_{2_{SB}}$ zero eigenvalues. \\
\textbf{(ii)} If the digraph has $l_1$ root nodes, they have no incoming edges, meaning the rows associated with the roots in the in-incidence matrix have only zeros, while all other nodes have at least one incoming edge. As a result, $\mbox{rank}(E_{s \odot}) = N-l_1$. If $l_1=0$, then, $\mbox{rank}(E_{s \odot}) = N$. \\
\textbf{(iii)} The statement follows from Item (ii) of Lemma \ref{lemma3a} and Item (iii) of Lemma \ref{lemma3b}. \\
\textbf{(iv)} From \cite[Section 0.4.5]{Horn}, it follows that $N-(l_1 + l_{2_{SB}}) \leq \mbox{rank}(L_{{s}}) \leq \min \{\mbox{rank}(E_{s \odot}), \mbox{rank}(E_{s}^\top)\}$. From (i), $L_s$ has exactly $N-(l_1 + l_{2_{SB}})$ nonzero eigenvalues. Moreover, from \cite[Lemma 1]{csekerciouglu2024distributed} we have that for $L_s$, the zero eigenvalue is semisimple, meaning its algebraic multiplicity equals its geometric multiplicity. Therefore the statement follows.\\ 
%\textbf{(iv)} From (i), we know that $L_s$ has $N-(l_1 + l_{2_{SB}})$ nonzero eigenvalues. So, $N-(l_1 + l_{2_{SB}}) \leq \mbox{rank}(L_{{s}}) \leq \min \{\mbox{rank}(E_{s \odot}), \mbox{rank}(E_{s}^\top)\}$ \cite[0.4.5]{Horn}, and the statement follows. \\
\textbf{(v)} First, from the proof of Item (iii) of Lemma \ref{lemma3a}, we know that $L_{s}$ and $L_{e_{s}}$ have the same nonzero eigenvalues, so 
$$N-(l_1 + l_{2_{\mathrm{SB}}}) \leq \mbox{rank}(L_{e_{s}}).$$
In addition, from \cite[0.4.5]{Horn}, we have 
\begin{align*}
    \mbox{rank}(E_{s}^\top) + \mbox{rank}(E_{s \odot}) - N &\leq \mbox{rank}(L_{e_{s}}) \\ 
    &\leq \min \{\mbox{rank}(E_{s}^\top), \mbox{rank}(E_{s \odot})\}.
\end{align*}
Then, for the case of SB digraphs, from (ii) and (iii), we have $\mbox{rank}(E_{s \odot}) = N-l_1$ and $\mbox{rank}(E_{s}) = N-1$. We also have that $l_{2_{\mathrm{SUB}}}=0$. For $l_1=0$, we have $N - 1 \leq \mbox{rank}(L_{e_{s}}) \leq N-1$, and (a) follows. For $l_{2_{\mathrm{SB}}}=0$, we have $N- l_1 \leq \mbox{rank}(L_{e_{s}}) \leq N-l_1$, and (b) follows. For the case of SUB digraphs, from (ii) and (iii), we have $\mbox{rank}(E_{s \odot}) = N-l_1$ and $\mbox{rank}(E_{s}) = N$. Then, we have $N-l_1 \leq \mbox{rank}(L_{e_{s}}) \leq N-l_1$, so $\mbox{rank}(L_{e_{s}})=N-l_1$, and if $l_1=0$, $\mbox{rank}(L_{e_{s}})=N$. Thus the statement follows.\\
\textbf{(vi)} In the case $l_1=0$, we have $\mbox{rank}({E_{s}}) = \mbox{rank}(L_{e_{s}})=N-1$ if the digraph is SB, and $\mbox{rank}({E_{s}}) = \mbox{rank}(L_{e_{s}})=N$, otherwise. Since $\mbox{rank}(E_{s \odot}) = N$, we can deduce that $\mathcal{N}(L_{e_s}^\top) = \mathcal{N}(E_s)$. In the case $l_1 \geq 1$, on the one hand, for SB digraphs, we have $N - l_1 - 1 \leq \mbox{rank}(L_{e_{s}}) \leq N-l_1$. For SUB digraphs, we have $\mbox{rank}(L_{e_{s}})=N-l_1$. On the other hand, $\mbox{rank}(E_{s}) = N-1$ for SB digraphs and $\mbox{rank}(E_{s}) = N$ for SUB digraphs, which indicates $\mathcal{N}(L_{e_s}^\top) \neq \mathcal{N}(E_s)$. Thus the statement follows.
\end{proof}

\begin{table*}[t]
    \centering
    \renewcommand{\arraystretch}{1.25}
    \caption{Geometric ($\gamma$) and algebraic ($\xi$) multiplicities of the zero eigenvalues of the directed signed edge Laplacian $L_{e_s}$. Note that $\gamma = \mathrm{dim}(\mathrm{ker}(L_{e_s})) = M-\mathrm{rank}(L_{e_s})$.} %The calculations/derivation of the geometric multiplicities follow from the fact]/OR just : [Note] that $\gamma = \mathrm{dim}(\mathrm{ker}(L_{e_s})) = M-\mathrm{rank}(L_{e_s})$.
    \begin{tabular}{|c|c|c|c|c|} %[Horn, Th. 2.4.11.1.b]
        \hline
         \textbf{Type} & \begin{tabular}{c} \textbf{Leader nodes} \end{tabular} & $\gamma$ & $\xi$ & Proof\\
         \hline \hline
         \multirow{4}{*}{SB} &  $l_1 = l_{2_{SB}} = 0$, or $l_1 =1$, or $l_{2_{SB}} =1$  & $M-N+1$  & $M-N+1$ & Item (iii) of Lemma \ref{lemma3a} \\
         \cline{2-5}
          & $l_1 > 1,\ l_{2_{SB}}=0$ & $M-N+l_1$ & $M-N+l_1$ & Items (iv)-(v) of Lemma \ref {lemma_multipleleaders}\\
         \cline{2-5}
          & $l_1 = 0,\ l_{2_{SB}} > 1$ & $M-N+1$ & $M-N+l_{2_{SB}}$ & Items (iv)-(v) of Lemma \ref {lemma_multipleleaders} \\
         \cline{2-5}
          & $l_1 \geq 1,\ l_{2_{SB}} \geq 1$ & $M-N+l_1$ & $M-N+l_1+l_{2_{SB}}$ & Items (iv)-(v) of Lemma \ref {lemma_multipleleaders} \\
         \hline
        \multirow{5}{*}{SUB}  & $l_1 = 1$, $l_{2_{SB}}=l_{2_{SUB}} = 0$ & $M-N+1$ & $M-N+1$ & Item (iv) of Lemma \ref{lemma3b}\\
        \cline{2-5}
         & $l_{2_{SB}} = 1,$ $l_1 = l_{2_{SUB}} =0$ & $M-N$ & $M-N+1$ & Item (v) of Lemma \ref{lemma3b}\\ %&\begin{tabular}{c}  $v_l^\top v_r = 1$ iff $M=N$ \\ $v_l^\top v_r = 0$ otherwise \end{tabular}
        \cline{2-5}
         & $l_{2_{SUB}} = 1,$  $l_1 = l_{2_{SB}}= 0$ & $M-N$ & $M-N$ & Item (iv) of Lemma \ref{lemma3b}\\
        \cline{2-5}
         & $l_{2_{SB}} = 0$, $l_1,l_{2_{SUB}} \geq 0$  & $M-N+l_1$ & $M-N+l_1$ & Items (iv)-(v) of Lemma \ref {lemma_multipleleaders} \\
        \cline{2-5}
         &  $l_{2_{SB}} > 0$, $l_1,l_{2_{SUB}} \geq 0$  & $M-N+l_1$ & $M-N+l_1+l_{2_{SB}}$ & Items (iv)-(v) of Lemma \ref {lemma_multipleleaders} \\
    \hline \end{tabular}
    \label{tab:eigenvalues}
\end{table*}
\subsection{On the eigenvectors of signed edge Laplacians}\label{section3ev}
We know that for signed digraphs containing a directed spanning tree (Lemmata \ref{lemma3a} and \ref{lemma3b}), and for signed digraphs with multiple leader groups (Lemma \ref{lemma_multipleleaders}), the signed edge Laplacian generally has multiple zero eigenvalues. Consequently, there exist multiple right and left eigenvectors associated with each zero eigenvalue of $L_{e_s}$. Moreover, depending on the structure of the signed digraph, these zero eigenvalues can have different geometric ($\gamma$) and algebraic ($\xi$) multiplicities. This affects both the Jordan decomposition of $L_{e_s}$ and the properties of its associated eigenvectors. The algebraic and geometric multiplicity properties of signed edge Laplacians of signed graphs containing a directed spanning tree and a root node or a rooted SCC, were presented in \cite{du2019further}. In Table \ref{tab:eigenvalues}, we extend these results to a broader class of signed digraphs, including those with multiple leader groups, and provide a unified characterization of the algebraic and geometric multiplicities of the zero eigenvalues.
%summarize the geometric and algebraic multiplicities of the zero eigenvalues for \textcolor{blue}{a larger class of signed digraphs}, including the ones containing multiple leader groups. 
These properties are crucial for defining the edge states and the average edge system, which in turn are used to characterize synchronization errors and analyze the convergence behavior of agents in the OMAS in Section \ref{section5}.

Based on the information in Table \ref{tab:eigenvalues} and following the framework in \cite{Horn}, the Jordan canonical form of $L_{e_s}$ can be expressed as
$$L_{e_s} = U J U^{-1},$$ 
where $U = [v_{r_1} \ v_{r_2} \ \cdots \ v_{r_M}] \in \mathbb{C}^{M \times M}$ contains the right eigenvectors and $U^{-1} = [v_{l_1}^\top \ v_{l_2}^\top \ \cdots \ v_{l_M}^\top]^\top \in \mathbb{C}^{M \times M}$ contains the left eigenvectors. In the matrix $J =\mbox{blkdiag} [J_0,\ \bar J]$, $\bar J$ corresponds to the Jordan blocks associated with eigenvalues having positive real parts, and $J_0$ contains the Jordan blocks corresponding to the zero eigenvalues, and is given by
\begin{align}\label{jordan}
J_0 = \mathrm{blkdiag}\left(
\underbrace{
\begin{bmatrix} 0 & 1 \\ 0 & 0 \end{bmatrix},\ \dots, \ \begin{bmatrix} 0 & 1 \\ 0 & 0 \end{bmatrix}
}_{\xi-\gamma\ \text{Jordan blocks}},
\
\underbrace{
0,\ \dots, \ 0
}_{\substack{2\gamma-\xi\\ \text{zeros}}}
\right).
\end{align}

From \eqref{jordan}, the right and left eigenvectors associated with the zero eigenvalues of $L_{e_s}$ satisfy the following:
\begin{equation}\label{edge_ev}
\begin{aligned}
    L_{e_s} v_{r_1} &= 0 \quad &v_{l_1}^\top L_{e_s} &= v_{l_2}^\top\\
    L_{e_s} v_{r_2} &= v_{r_1} \quad &v_{l_2}^\top L_{e_s} &= 0\\
    &\vdots \quad & &\vdots\\
    L_{e_s} v_{r_{2(\xi-\gamma)-1}} &= 0 \quad &v_{l_{2(\xi-\gamma)-1}}^\top L_{e_s} &= v_{l_{2(\xi-\gamma)}}^\top\\
    L_{e_s} v_{r_{2(\xi-\gamma)}} &= v_{r_{2(\xi-\gamma)-1}} \quad &v_{l_{2(\xi-\gamma)}}^\top L_{e_s} &= 0\\
    L_{e_s} v_{r_{2(\xi-\gamma)+1}} &= 0 \quad &v_{l_{2(\xi-\gamma)+1}}^\top L_{e_s} &= 0\\
    &\vdots \quad & &\vdots\\
    L_{e_s} v_{r_{\xi}} &= 0 \quad &v_{l_{\xi}}^\top L_{e_s} &= 0.
\end{aligned}
\end{equation}
From Table \ref{tab:eigenvalues}, it is clear that whenever the graph contains at least one SB SCC in the graph, the geometric and algebraic multiplicities of the zero eigenvalue differ. As a consequence, for the first $2(\xi-\gamma)$ left and right eigenvectors, the product with $L_{e_s}$ is nonzero for every other pair, while the rest yield zero. On the other hand, the product involving the remaining $2\gamma - \xi$ eigenvectors is zero.

\begin{remark}\label{rmk:identity}
Throughout the paper, the left and right eigenvectors associated with the zero eigenvalues of $L_{e_s}$ are normalized such that $v_{l_i}^\top v_{r_i}=1$ for all $i \leq \xi$.
\end{remark}

\begin{remark} Recall from \eqref{498} that the signed Laplacian and signed edge-Laplacian can be written as $L_{e_s} = E_s^\top E_{s\odot} \in \mathbb{R}^{M\times M}$ and $L_s=E_{s\odot}E_s^\top \in \mathbb{R}^{N\times N}$, respectively. From the proof of Item (iii) of Lemma \ref{lemma3a}, we know that $L_{e_s}$ and $L_{s}$ have the same nonzero eigenvalues. From \cite[Lemma 1]{csekerciouglu2024distributed} we have that for $L_s$, the zero eigenvalue is semisimple, meaning its algebraic multiplicity equals its geometric multiplicity. Then, the size of every Jordan block corresponding to the zero eigenvalue is exactly $1\times1$. Let $A=E_s^\top$ and $B=E_{s\odot}$. It is well known that the sizes of the Jordan blocks associated with the zero eigenvalue of $AB$ and $BA$ may differ by at most one \cite{lippert2009jordan}. Then, this is also the case for $L_{e_s}$ and $L_{s}$. Consequently, the nilpotency index of the zero eigenvalue of $L_{e_s}$ is at most two, and the Jordan blocks associated with $\lambda=0$ are necessarily of size $1\times1$ or $2\times2$. This justifies the Jordan structure given in \eqref{edge_ev}.\end{remark}

\subsection{Lyapunov Equation for Directed Signed Edge Laplacians}\label{section4}
To establish synchronization of multi-agent systems over signed digraphs, we prove the exponential stability of the set $\{ \bar e = 0\}$, where $\bar e$ is defined in \eqref{def_e_bar}. In particular, we remind the results on how to construct strict Lyapunov functions, in the space of $\bar e$, referring to functions expressed in terms of $\bar e$, presented in \cite[Theorems 2--4]{TACPelin}.

For a signed digraph containing a directed spanning tree, we have the following.

\begin{thm} \label{thm:lyap_eq_Le_directed} (\cite[Theorem 3]{TACPelin}) Let $\mathcal{G}_s$ be a signed digraph of $N$ agents interconnected by $M$ edges. Let $L_{e_{s}}$ be the associated directed edge Laplacian, which contains $\xi$ zero eigenvalues. Then, the following are equivalent:
	\renewcommand{\theenumi}{(\roman{enumi}}
	\begin{enumerate}
		\item $\mathcal{G}_s$ contains a directed spanning tree,
        \item for any $Q \in \mathbb{R}^{M\times M}, Q=Q^{\top}> 0$ and for any $\{ \alpha_{1}, \alpha_{2}, \dots, \alpha_{\xi}\}$ with $\alpha_{i} >0$, there exists a matrix $P(\alpha_{i}) \in \mathbb{R}^{M \times M}$, $P=P^{\top}>0$ such that 
\begin{equation}\label{Lyap_eq_dir}
    PL_{e_{s}} + L_{e_{s}}^{\top}P = Q - \sum_{i=1}^{\xi} \alpha_{i} \left(Pv_{r_{i}}v_{l_{i}}^{\top} + v_{l_{i}}v_{r_{i}}^{\top}P\right),
\end{equation}
where $v_{r_{i}}, v_{l_{i}} \in \mathbb{R}^M$ are, respectively, the right and left eigenvectors of $L_{e_{s}}$ associated with the $i$th 0 eigenvalue. 
\end{enumerate}
$\xi$ is given in Table \ref{tab:eigenvalues} and satisfies $\xi = M-N+1$ if the signed digraph is SB, and also if it is SUB with a root node or SB-rooted SCC; otherwise, $\xi = M-N$.
%Moreover, if the signed digraph is SB, then $\xi = M-N+1$. If the signed digraph is SUB, then $\xi = M-N+1$ if and only if it contains either a root node or a SB-rooted SCC; otherwise, $\xi = M-N$.
\end{thm} 

\begin{corol}\label{cor:lyap_eq_spanningtree} (\cite[Theorem 2]{TACPelin})
Let $\mathcal{G}_s$ be a signed directed spanning tree. Then the associated edge Laplacian $L_{e_s}$ has no zero eigenvalues, i.e., $\xi = 0$, and for any
$Q \in \mathbb{R}^{(N-1) \times (N-1)}$, $Q = Q^\top > 0$, there exists a matrix
$P \in \mathbb{R}^{(N-1) \times (N-1)}$, $P = P^\top > 0$, such that
\begin{equation}\label{Lyap-eq}
    PL_{e_s} + L_{e_s}^\top P = Q.
\end{equation}
\end{corol}

\begin{thm} \label{proposition:lyap_eq} (\cite[Theorem 4]{TACPelin})
    Let $\mathcal{G}_s$ be a signed digraph of $N$ agents interconnected by $M$ edges, containing $m$ leader groups. Let $l_1$ be the number of root nodes, $l_{2_{SB}}$ be the number of SB-rooted SCCs, and $l_{2_{SUB}}$ be the number of SUB-rooted SCCs, where $m=l_1 + l_2$ and $l_2 = l_{2_{SB}} +l_{2_{SUB}}$. Let $L_{e_{s}}$ be the associated directed edge Laplacian. Then the following are equivalent:
\renewcommand{\theenumi}{(\roman{enumi}}
\begin{enumerate}
	\item the graph has $m$ leader groups, as defined in Assumption \ref{ass2}, and given each follower $\nu_j, \forall j \in \mathcal{F}$, there exists at least one leader $\nu_i, \forall i \in \mathcal{L}$, such that there exists at least one path from $\nu_i$ to $\nu_j$,
	\item for any $Q \in \mathbb{R}^{M \times M}, Q=Q^{\top}> 0$ and for any $\{ \alpha_1, \alpha_2, \dots, \alpha_{\xi}\}$ with $\alpha_i >0$, there exists a matrix $P(\alpha_i) \in \mathbb{R}^{M \times M}$, $P=P^{\top}>0$ such that \eqref{Lyap_eq_dir} holds,
where $v_{ri}, v_{li} \in \mathbb{R}^M$ are the right and left eigenvectors of $L_{e_{s}}$ associated with the $i$th 0 eigenvalue.
\end{enumerate}
Moreover, $\xi$ satisfies $\xi = M-N+l_1+l_{2_{SB}}$ whether the signed digraph is SB or SUB.
\end{thm}

\section{Exponential Stability}\label{section5}
In this section, we present our main result. We consider the synchronization problem of multi-agent systems \eqref{FO}-\eqref{CL} over signed digraphs, and establish the exponential stability of the synchronization errors. To this end, we first express the synchronization errors using the edge-Laplacian notation introduced in Section \ref{section3}, and then reformulate the control problem accordingly. 

Using the definition of the incidence matrix in \eqref{434}, we may express the edge states in \eqref{def_e} in vector form
\begin{align}\label{err_vect}
	e = E_{s}^\top x,
\end{align}
where $E_{s}$ is the incidence matrix corresponding to the graph. For signed digraphs, from the definition of the directed Laplacian in \eqref{498}, we write the control law in \eqref{CL} in vector form as
\begin{align}\label{CL_vect_edge}
	u = -k_1  E_{{s \odot}} {e}.
\end{align}
Differentiating the edge states \eqref{err_vect} yields
\begin{align}\label{edgestates}
	\dot{e}= -k_1 E_{s}^\top E_{{s \odot}} e =  -k_1 L_{e_{s}} e.
\end{align}
Similarly, differentiating the synchronization errors \eqref{def_e_bar},
\begin{equation}
	\dot{\bar{e}} = [I -\sum_{i=1}^{\xi} v_{r_{i}}v_{l_{i}}^{\top} ] \dot{e} = -k_1 [I -\sum_{i=1}^{\xi} v_{r_{i}}v_{l_{i}}^{\top} ] L_{e_{s}}{e},
\end{equation}
we obtain
\begin{align}\label{dyn_edge1}
  \dot{\bar e} = -k_1 L_{e_{s}} \bar e.
\end{align}
\begin{remark}\label{explanation}
Observe that, to obtain \eqref{dyn_edge1}, we distinguish two cases based on Table \ref{tab:eigenvalues} and the relations in \eqref{edge_ev}: the first case where $\xi = \gamma$, and the second where $\xi>\gamma$. Here, $\gamma$ and $\xi$ denote the geometric and algebraic multiplicities of the zero eigenvalue.
\begin{itemize}
    \item Case 1 ( $\xi = \gamma$ ): In this case, all left eigenvectors satisfy $v_{l_{i}}^{\top}L_{e_s}=0.$ Consequently, we have $\dot{\bar{e}} = -k_1L_{e_{s}}\bar{e}.$ Additionally, since $L_{e_s}v_{r_{i}}=0$, it also holds that $\dot{\bar{e}} = -k_1L_{e_{s}}[I -\sum_{i=1}^{\xi} v_{r_{i}}v_{l_{i}}^{\top} ]{e}.$ Finally, substituting \eqref{def_e_bar} in the latter, we obtain \eqref{dyn_edge1}.
    \item Case 2 ( $\xi>\gamma$ ): For $i \in \{1, 2, \dots, \xi-\gamma \}$, we have the relations $v_{l_{{2i-1}}}^{\top}L_{e_s}=v_{l_{{2i}}}^{\top},$ $v_{l_{{2i}}}^{\top}E_{s}^\top=0,$ and $L_{e_s}v_{r_{{2i}}}=v_{r_{{2i-1}}}$, $E_{s}v_{r_{{2i-1}}}=0$. The remaining $2\gamma - \xi$ right and left eigenvectors associated with zero eigenvalues satisfy $L_{e_s}v_{r_{i}}=0$ and $v_{l_{i}}^{\top}L_{e_s}=0.$ Thus, for this case as well, $\dot{\bar{e}} = -k_1L_{e_{s}}[I -\sum_{i=1}^{\xi} v_{r_{i}}v_{l_{i}}^{\top} ]{e}$ holds, and substituting \eqref{def_e_bar} in the latter, Eq. \eqref{dyn_edge1} follows.
\end{itemize}
\end{remark}
 % Then, the closed-loop system for the synchronization error dynamics is given as

We are now ready to present our main result, global exponential stability of the origin of the set $\{ \bar e = 0\}$.  %Let the control objective for \eqref{dyn_edge1} be to ensure that the origin is globally exponentially stable, that is,
% \begin{align}\label{obj_e}
% 	\lim_{t \to \infty} \bar e(t) = 0.
% \end{align}

\begin{thm}\label{prop:result2}
Consider the system \eqref{dyn_edge1}. Under Assumptions~\ref{standing_ass} and \ref{ass1} or Assumptions \ref{standing_ass} and \ref{ass2}, for any $Q=Q^\top>0$ there exists $P=P^\top$ such that 
    $$V(\bar e) := \frac{1}{2} \bar e^\top P \bar e >0 \text{ and } \dot V(\bar e) = - \frac{1}{2} k_1 \bar e^\top Q \bar e<0, \ \forall\bar e \neq 0.$$
    Then, the set $\{ \bar e = 0\}$ is globally exponentially stable. Furthermore, for a signed digraph, under Assumption \ref{ass1}:
    \renewcommand{\theenumi}{(\roman{enumi}}
	\begin{enumerate}
    \item If $\mathcal{G}_{s}$ is SB, then agents achieve bipartite consensus, that is, the inequality \eqref{obj_BC} holds. 
    \item If $\mathcal{G}_{s}$ is SUB and contains a SUB-rooted SCC without a root node, then agents achieve trivial consensus, that is, the inequality \eqref{obj_C} holds.
    \item If $\mathcal{G}_{s}$ is SUB and contains either a SB-rooted SCC or a root node, then agents achieve interval bipartite consensus, that is, the inequality \eqref{obj_IntBC} holds.
    \end{enumerate}
    On the other hand, for a signed digraph containing $m$ leader groups, under Assumption \ref{ass2}:
    \begin{enumerate}
     \item[(iv)] If $\mathcal{G}_{s}$ contains at least one root node or a SB-rooted SCC, then agents achieve bipartite containment; that is, if $\mathcal{G}_{s}$ is SB, then inequality \eqref{eq:containment_SB} holds, whereas if $\mathcal{G}_{s}$ is SUB, inequality \eqref{eq:containment_SUB} holds.
     \item[(v)] If $\mathcal{G}_{s}$ contains only SUB-rooted SCCs, then agents achieve trivial consensus, that is, the inequality \eqref{obj_C} holds.
    \end{enumerate}
\end{thm}

\begin{proof}
Let $Q = Q^\top > 0$ and $\alpha>0$ be arbitrarily fixed. By Assumption \ref{ass1} and Theorem \ref{thm:lyap_eq_Le_directed}, there exists a symmetric positive definite matrix $P = P^\top > 0$ such that \eqref{Lyap_eq_dir} holds. In the case where the considered graph is a spanning tree, by Corollary \ref{cor:lyap_eq_spanningtree}, there exists a $P = P^\top > 0$ such that \eqref{Lyap-eq} holds. In the case the digraph contains multiple leader groups, by Assumption \ref{ass2} and Theorem \ref{proposition:lyap_eq}, there exists a $P = P^\top > 0$ such that \eqref{Lyap_eq_dir} holds. Then, consider the Lyapunov function candidate $V(\bar e) := \frac{1}{2} \bar e^\top P \bar e.$ Its total time derivative along the trajectories of \eqref{dyn_edge1} yields $\dot V(\bar e) = - k_1 \bar e^\top P L_{e_{s}} \bar e.$  If the underlying signed digraph is a directed spanning tree, using \eqref{Lyap-eq}, we obtain $\dot V(\bar e) = - \frac{1}{2} k_1 \bar e^\top Q \bar e.$ If the underlying signed digraph contains a directed spanning tree or multiple leader groups, using \eqref{Lyap_eq_dir}, we obtain $	\dot V(\bar e) = - \frac{1}{2} k_1 \bar e^\top Q \bar e + \frac{1}{2} \bar e^\top \left( \sum_{i=1}^{\xi} \alpha_{i} (Pv_{r_{i}}v_{l_{i}}^{\top} + v_{l_{i}}v_{r_{i}}^{\top}P) \right) \bar e.$ From the definition of the synchronization errors \eqref{def_e_bar} and using the identity $v_{l_{i}}^{\top}v_{r_{i}}=1,\ i \leq \xi$ (Remark \ref{rmk:identity}), we have $\sum_{i=1}^{\xi} \alpha_{i} Pv_{r_{i}}v_{l_{i}}^{\top} [I - \sum_{i=1}^{\xi} v_{r_{i}}v_{l_{i}}^{\top}]e
        = \sum_{i=1}^{\xi} \alpha_{i} Pv_{r_{i}}v_{l_{i}}^{\top}e - \sum_{i=1}^{\xi} \alpha_{i} Pv_{r_{i}}v_{l_{i}}^{\top} e = 0$
    and $\bar e^\top [I - \sum_{i=1}^{\xi} v_{r_{i}}v_{l_{i}}^{\top}]^\top \sum_{i=1}^{\xi}\alpha_{i} v_{l_{i}}v_{r_{i}}^{\top}P
        = \sum_{i=1}^{\xi} e^\top  \alpha_{i} v_{l_{i}}v_{r_{i}}^{\top}P - \sum_{i=1}^{\xi} e^\top  \alpha_{i} v_{l_{i}}v_{r_{i}}^{\top}P = 0.$ Thus, $\dot V(\bar e) = - \frac{1}{2} k_1 \bar e^\top Q \bar e < - \frac{1}{2} k_1 \lambda_{\min}(Q) | \bar e |^2 < 0,$ in both cases, where $\lambda_{\min}(Q) >0$ is the smallest eigenvalue of $Q$. The first statement of the proposition follows.

Then, differentiating the weighted average system \eqref{e_m}, we obtain the dynamical equation $\dot{e}_m = \sum_{i=1}^{\xi}v_{r_i}v_{l_i}^\top \dot{e} = - k_1 \sum_{i=1}^{\xi}v_{r_i}v_{l_i}^\top L_{e_s} e =0$ (see Remark \ref{explanation}). Its solution gives $e_m(0) = e_m(t)$. Moreover, from the previous result, we have $\lim_{t \to \infty} \bar e (t)=0$, which gives, from \eqref{def_e_bar}, $\lim_{t \to \infty} e (t)= e_{m}(t)=  e_{m}(0)$. Then, from \eqref{e_m} and \eqref{err_vect}, we obtain
$		\lim_{\bar t \to \infty} e(t) = \left[ v_{r_{1}}v_{l_{1}}^{\top} + \cdots + v_{r_{{\xi}}}v_{l_{{\xi}}}^{\top}\right] e(0)
		= \left[ v_{r_{1}}v_{l_{1}}^{\top} + \cdots + v_{r_{{\xi}}}v_{l_{{\xi}}}^{\top}\right] E_{s}^\top {x}(0),$
	where $v_{r_{i}}$ and $v_{l_{i}}$ are the right and left eigenvectors associated with the $\xi$ zero eigenvalues of $L_{e_{s}}$.

\noindent For a signed digraph containing a spanning tree:
    
   \textbf{(i)} In the case the signed digraph is SB, $\xi = M - N + 1$, and since, from Lemma \ref{NE_SUB}, $v_{l_{1}}^{\top} E_{s}^\top = v_{l_{2}}^{\top} E_s^\top = \cdots = v_{l_{{\xi}}}^{\top} E_{s}^\top = 0$, it leads to $\lim_{t \to \infty} E_{s}^\top {x}(t)= 0.$ Then, using the edge-gauge transformation from Lemma \ref{lemma_edgegauge}, we have that $\lim_{t\to \infty} E_{s}^\top {x}(t)= \lim_{t \to \infty} D_{e} E^\top D {x}(t) = 0.$ From \cite[Theorem 3.4]{zelazo2007agreement}, the null space of $E^\top$ is spanned by $\mathbf{1}_{N}$, so $\lim_{t \to \infty} D {x}(t) = c \mathbf{1}_{N},$ where $c \in \mathbb{R}$. Therefore, we can deduce that $\lim_{t \to \infty} {x}(t) = \alpha D \mathbf{1}_{N},$ where $\alpha \in \mathbb{R}$ is a constant determined by the initial conditions of the agents and the signed interactions. This implies that the system achieves bipartite consensus.

	\textbf{(ii)} In the case that the signed graph is SUB and contains a SUB-rooted SCC without a root node, from Item (i) of Lemma \ref{NE_SUB}, we have that $v_{l_{1}}^{\top} E_{s}^\top = v_{l_{2}}^{\top} E_{s}^\top = \cdots = v_{l_{{\xi}}}^{\top} E_{s}^\top = 0$, it leads to $\lim_{t \to \infty} E_{s}^\top {x}(t)= 0.$ Since from Item (iii) of Lemma \ref{lemma3b}, $\mbox{rank}(E_{s}) = N$, the only solution is 
	$\lim_{t\to \infty}{x}(t) = 0,$ which implies that the system achieves trivial consensus.

    \textbf{(iii)} In the case that the signed graph is SUB and contains a SB-rooted SCC, from Item (i) of Lemma \ref{NE_SUB}, we have that $\mathcal{N}(L_{e_s}^\top) = \mathcal{N}(E_s)$. On the other hand, from Table \ref{tab:eigenvalues}, the algebraic multiplicity of the zero eigenvalue is $M - N + 1$, while its geometric multiplicity is $M - N$. Then, from Eq. \eqref{edge_ev}, we have $L_{e_s}v_{r_2}=v_{r_1}$ and $v_{l_1}^\top L_{e_s}=v_{l_2}^\top$. Thus, $v_{l_{{1}}}^{\top} E_{s}^\top \neq 0$, and $\lim_{t \to \infty} e(t) = v_{r_{1}}v_{l_{1}}^{\top} E_{s}^\top {x}(0).$ In the case that the signed graph is SUB and contains a root node, from Item (ii) of Lemma \ref{NE_SUB}, we have that $\mathcal{N}(L_{e_s}^\top) \neq \mathcal{N}(E_s)$ so $v_{l_{{k}}}^{\top} E_{s}^\top \neq 0$ for all $k \leq \xi$. Moreover, we also have from Item (iii) of Lemma \ref{lemma3b} that $\mbox{rank}(E_{s}) = N$ so $\lim_{t \to \infty} e(t) \neq 0$. Then,
    $\lim_{t \to \infty} e(t) = \sum_{i=1}^{\xi} v_{r_{i}}v_{l_{i}}^{\top} E_{s}^\top {x}(0).$ 
    Thus, the only achievable objective is interval bipartite consensus.

\noindent For a signed digraph containing multiple leaders, $L_{e_s}$ contains $M - N + l_1 + l_{2_{SB}}$ zero eigenvalues. 

\textbf{(iv)} First, suppose that the signed digraph contains at least one root node. From Item (vi) of Lemma~\ref{lemma_multipleleaders}, we have $\mathcal{N}(L_{e_s}^\top) \neq \mathcal{N}(E_s)$. Let $\gamma$ denote the geometric multiplicity of the zero eigenvalue. Then, from Table \ref{tab:eigenvalues} and the relations in \eqref{edge_ev}, we obtain
$    \lim_{t \to \infty} e(t) = \left(\sum_{k = 0}^{\gamma - \xi - 1} v_{r_{{2k + 1}}} v_{l_{{2k + 1}}}^{\top} + \sum_{i = 2(\gamma - \xi)}^{\xi} v_{r_{i}} v_{l_{i}}^{\top} \right) E_{s}^\top {x}(0).$
Now, suppose that all leader groups are rooted SCCs. In this case, $L_{e_s}$ has $M - N + l_{2_{SB}}$ zero eigenvalues. On the one hand, Lemma~\ref{lemma_multipleleaders} implies $\mathcal{N}(L_{e_s}^\top) = \mathcal{N}(E_s)$. On the other hand, from Table \ref{tab:eigenvalues}: (a) If the graph is SB, then the algebraic multiplicity of the zero eigenvalue is $M - N + l_{2_{SB}}$, while its geometric multiplicity is $M - N + 1$. (b) If the graph is SUB, but contains at least one SB-rooted SCC, then the algebraic and geometric multiplicities are $M - N + l_{2_{SB}}$ and $M - N$, respectively. Consequently, from \eqref{jordan} and \eqref{edge_ev}, 
$ \lim_{t \to \infty} e(t) =  \sum_{k = 0}^{\gamma - \xi -1} v_{r_{{2k + 1}}} v_{l_{{2k + 1}}}^{\top}  E_{s}^\top {x}(0).$ Thus, following the results in \cite{sekercioglu2023exponential}, the only achievable objective for the agents is bipartite containment.

\textbf{(v)} Finally, consider the case in which the signed digraph contains only SUB-rooted SCCs. Then, both the algebraic and geometric multiplicities of the zero eigenvalue are equal to $M - N$, and $v_{l_{1}}^{\top} E_{s}^\top = v_{l_{2}}^{\top} E_{s}^\top = \cdots = v_{l_{_{\xi}}}^{\top} E_{s}^\top = 0$. Hence, $\lim_{t \to \infty} E_{s}^\top {x}(t)= 0.$ Since, from Item (iii) of Lemma~\ref{lemma_multipleleaders}, $\mbox{rank}(E_{s}) = N$, the only possible solution is 
$\lim_{t\to \infty}{x}(t) = 0,$
which implies that the system achieves trivial consensus.
\end{proof}

\section{Numerical Example}\label{section6}
\begin{figure}[t!]\centering
\begin{subfigure}[b]{0.1\textwidth}
			\centering
\begin{tikzpicture}[node distance={8mm}, thick,main/.style = {draw, circle}] 
	\tikzset{mynode/.style={circle,draw,minimum size=12pt,inner sep=0pt,thick},}
		\node[mynode] (1) {$\nu_1$}; 
		\node[mynode] (2) [right  of=1] {$\nu_2$};
		\node[mynode] (3) [below of=1] {$\nu_3$};
		\node[mynode] (4) [below of=2] {$\nu_4$};
		\node[mynode, fill=blue, fill opacity=0.2, draw=blue, text opacity=1] (5) [above of=1] {$\nu_5$};
		\draw[->, color=red, dash pattern=on 1mm off 1mm] (1) -- node[midway, above] {$e_1$}(2);
		\draw[->] (1) -- node[midway, left] {$e_2$}(3);
		\draw[->] (2) -- node[midway, right] {$e_3$}(4);
		\draw[->] (3) -- node[midway, below] {$e_4$}(4);
		\draw[<-] (1) -- node[midway,left] {$e_5$}(5);
	\end{tikzpicture}
\caption{$\mathcal{G}_1$}
\end{subfigure}
\begin{subfigure}[b]{0.12\textwidth}
	\centering
\begin{tikzpicture}[node distance={8mm}, thick,main/.style = {draw, circle}] 
	\tikzset{mynode/.style={circle,draw,minimum size=12pt,inner sep=0pt,thick},}
    \definecolor{pastelgreen}{rgb}{0.4660 0.6740 0.1880}
	\node[mynode] (1) {$\nu_1$}; 
	\node[mynode] (2) [right  of=1] {$\nu_2$};
	\node[mynode] (3) [below of=1] {$\nu_3$};
	\node[mynode] (4) [below of=2] {$\nu_4$};
	\node[mynode] (5) [left of=1] {$\nu_5$};
	\node[mynode] (6) [left of=3] {$\nu_6$};
    \node[mynode, fill=blue, fill opacity=0.2, draw=blue, text opacity=1] (7) [above of=1, xshift=-0.5cm] {$\nu_7$};
	\node[mynode] (8) [above of=1, xshift=0.5cm] {$\nu_8$};
	\node[mynode, fill=blue, fill opacity=0.2, draw=blue, text opacity=1] (9) [below of=3] {$\nu_9$};
    \draw[red, dash pattern=on 1mm off 1mm][->](1) -- node[midway, below] {$e_1$}(2);
	\draw[->] (1) -- node[midway, right] {$e_2$}(3);
	\draw[->] (2) -- node[midway, right] {$e_3$}(4);
	\draw[red, dash pattern=on 1mm off 1mm][->] (3) -- node[midway, below] {$e_4$}(4);
	\draw[->] (1) -- node[midway, below] {$e_5$}(5);
	\draw[->, color=red, dash pattern=on 1mm off 1mm] (3) -- node[midway, above] {$e_6$}(6);
    \draw[->] (7) -- node[midway, above] {$e_8$}(8);
    \draw[->, color=red, dash pattern=on 1mm off 1mm] (8) -- node[midway, right] {$e_9$}(1);
    \draw[<-, color=red, dash pattern=on 1mm off 1mm] (1) -- node[midway, left] {$e_7$}(7);
    \draw[->] (9) -- node[midway, left] {$e_{10}$}(3);
	\end{tikzpicture}
	\caption{$\mathcal{G}_2$}
\end{subfigure}
\begin{subfigure}[b]{0.12\textwidth}
	\centering
	\begin{tikzpicture}[node distance={8mm}, thick,main/.style = {draw, circle}] 
     \definecolor{pastelgreen}{rgb}{0.4660 0.6740 0.1880}
\tikzset{mynode/.style={circle,draw,minimum size=12pt,inner sep=0pt,thick},}
	\node[mynode, fill=pastelgreen, fill opacity=0.2, draw=pastelgreen, text opacity=1] (1) {$\nu_1$}; 
	\node[mynode] (2) [right  of=1] {$\nu_2$};
	\node[mynode] (3) [below of=1] {$\nu_3$};
	\node[mynode] (4) [below of=2] {$\nu_4$};
	\node[mynode] (5) [left of=1] {$\nu_5$};
	\node[mynode] (6) [left of=3] {$\nu_6$};
    \node[mynode, fill=pastelgreen, fill opacity=0.2, draw=pastelgreen, text opacity=1] (7) [above of=1, xshift=-0.5cm] {$\nu_7$};
	\node[mynode, fill=pastelgreen, fill opacity=0.2, draw=pastelgreen, text opacity=1] (8) [above of=1, xshift=0.5cm] {$\nu_8$}; %
	\node[mynode, fill=blue, fill opacity=0.2, draw=blue, text opacity=1] (9) [below of=3] {$\nu_9$};
    \draw[red, dash pattern=on 1mm off 1mm][->](1) -- node[midway, below] {$e_1$}(2);
	\draw[->] (1) -- node[midway, right] {$e_2$}(3);
	\draw[->] (2) -- node[midway, right] {$e_3$}(4);
	\draw[->, color=red, dash pattern=on 1mm off 1mm] (3) -- node[midway, below] {$e_4$}(4);
	\draw[->] (1) -- node[midway, below] {$e_5$}(5);
	\draw[->, color=red, dash pattern=on 1mm off 1mm] (3) -- node[midway, above] {$e_6$}(6);
    \draw[->] (7) -- node[midway, above] {$e_8$}(8);
    \draw[->] (8) -- node[midway, right] {$e_9$}(1);
    \draw[->] (1) -- node[midway, left] {$e_7$}(7);
    \draw[->] (9) -- node[midway, left] {$e_{10}$}(3);
	\end{tikzpicture}
	\caption{$\mathcal{G}_3$}
\end{subfigure}
\begin{subfigure}[b]{0.12\textwidth}
	\centering
	\begin{tikzpicture}[node distance={8mm}, thick,main/.style = {draw, circle}] 
\tikzset{mynode/.style={circle,draw,minimum size=12pt,inner sep=0pt,thick},}
	\node[mynode, fill=red, fill opacity=0.2, draw=red, text opacity=1] (1) {$\nu_1$}; 
	\node[mynode] (2) [right  of=1] {$\nu_2$};
	\node[mynode] (3) [below of=1] {$\nu_3$};
	\node[mynode] (4) [below of=2] {$\nu_4$};
	\node[mynode] (5) [left of=1] {$\nu_5$};
	\node[mynode] (6) [left of=3] {$\nu_6$};
    \node[mynode, fill=red, fill opacity=0.2, draw=red, text opacity=1] (7) [above of=1, xshift=-0.5cm] {$\nu_7$};
	\node[mynode, fill=red, fill opacity=0.2, draw=red, text opacity=1] (8) [above of=1, xshift=0.5cm] {$\nu_8$}; %
	\node[mynode, fill=blue, fill opacity=0.2, draw=blue, text opacity=1] (9) [below of=3] {$\nu_9$};
    \draw[red, dash pattern=on 1mm off 1mm][->](1) -- node[midway, below] {$e_1$}(2);
	\draw[->] (1) -- node[midway, right] {$e_2$}(3);
	\draw[->] (2) -- node[midway, right] {$e_3$}(4);
	\draw[->, color=red, dash pattern=on 1mm off 1mm] (3) -- node[midway, below] {$e_4$}(4);
	\draw[->] (1) -- node[midway, below] {$e_5$}(5);
	\draw[->, color=red, dash pattern=on 1mm off 1mm] (3) -- node[midway, above] {$e_6$}(6);
    \draw[->] (7) -- node[midway, above] {$e_8$}(8);
    \draw[->] (8) -- node[midway, right] {$e_9$}(1);
    \draw[->, color=red, dash pattern=on 1mm off 1mm] (1) -- node[midway, left] {$e_7$}(7);
    \draw[->] (9) -- node[midway, left] {$e_{10}$}(3);
	\end{tikzpicture}
	\caption{$\mathcal{G}_4$}
\end{subfigure}
  \caption{The black edges represent cooperative interactions, and the dashed red edges represent antagonistic interactions. (a) $\mathcal{G}_1$ is SUB and has $\nu_5$ as a root node. (b) $\mathcal{G}_2$ is SB ($\mathcal{V}_1 = \{ \nu_1, \nu_3, \nu_5, \nu_9 \}$ and $\mathcal{V}_2 = \{ \nu_2, \nu_4, \nu_6, \nu_7, \nu_8 \}$) and has $\nu_7$ and $\nu_9$ as root nodes. (c) $\mathcal{G}_3$ is SB ($\mathcal{V}_1 = \{ \nu_1, \nu_3, \nu_5, \nu_7, \nu_8, \nu_9 \}$ and $\mathcal{V}_2 = \{ \nu_2, \nu_4, \nu_6 \}$) and contains $\nu_9$ as a root node and a SB-rooted SCC formed by $\nu_1$, $\nu_8$, and $\nu_9$. (d) $\mathcal{G}_3$ is SUB and contains $\nu_9$ as a root node and a SUB-rooted SCC formed by $\nu_1$, $\nu_8$, and $\nu_9$.}
  \label{graph}
\end{figure}
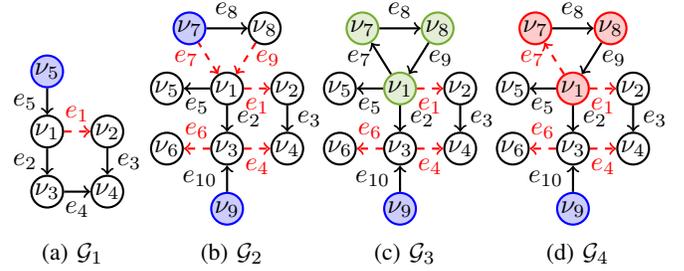

\begin{figure*}[h!]
\centering
\subfloat[]{
\includegraphics[width=0.95\columnwidth]{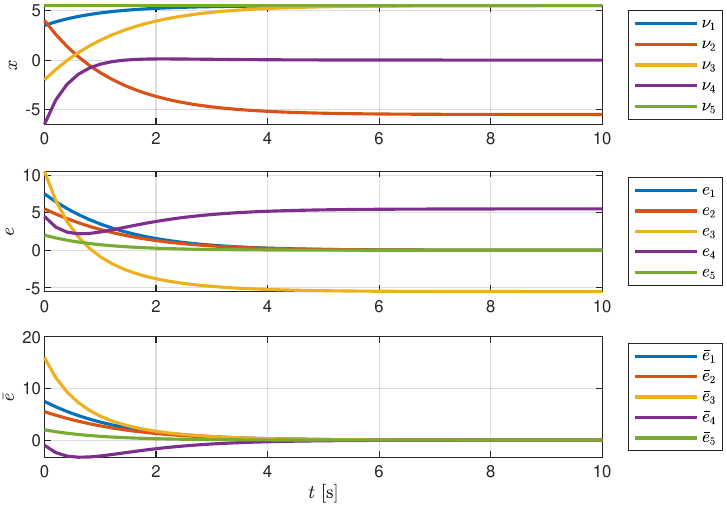}\label{fig1}}\quad
\subfloat[]{
\includegraphics[width=0.95\columnwidth]{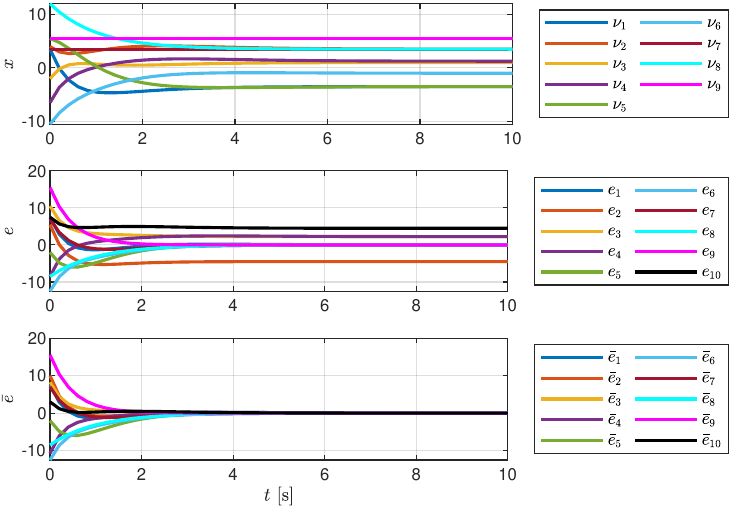}
\label{fig2}}\\[5pt]
\subfloat[]{\includegraphics[width=0.95\columnwidth]{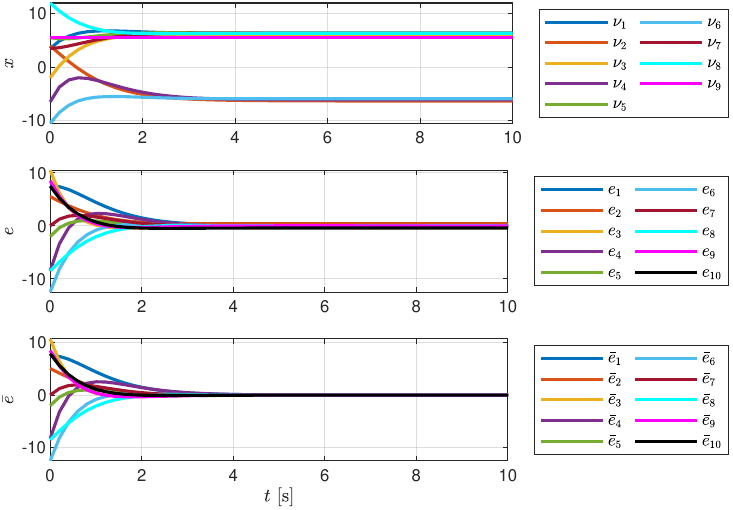}\label{fig3}}
\quad
\subfloat[]{\includegraphics[width=0.95\columnwidth]{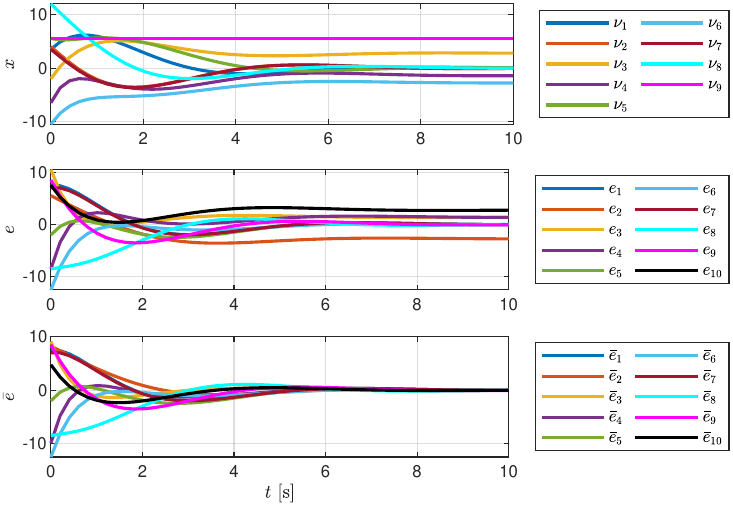}\label{fig4}}
\caption{ {\small Evolution of the trajectories of agents $x$ \eqref{FO}-\eqref{CL}, edges $e$ \eqref{edgestates} and synchronization errors $\bar e$ \eqref{dyn_edge1}. (a): $\mathcal{G}_1$.  (b): $\mathcal{G}_2$. (c): $\mathcal{G}_3$. (d): $\mathcal{G}_4$.}} 
\end{figure*}

We illustrate our theoretical findings by considering multi-agent systems evolving according to \eqref{FO}-\eqref{CL}, represented by four different signed digraphs depicted in Figure \ref{graph}. Let $k_1=4$ and $P$ generated by \eqref{Lyap-eq} or \eqref{Lyap_eq_dir}, depending on the considered graph structure, with $Q=I_{M\times M}$, $\alpha_k=1$, $k\leq 9$. The initial conditions of the agents for $\mathcal{G}_1$ are $[3.5,\ 4,\ -2,\ -6.5,\ 5.5]$ and for $\mathcal{G}_2,\ \mathcal{G}_3,$ and $\mathcal{G}_4$ are $[3.5,\ 4,\ -2,\ -6.5,\ 5.5,\ -10.5,\ 3.5,\ 12,\ 5.5]$. The trajectory evolution of the states of the agents \eqref{FO}-\eqref{CL}, edges \eqref{edgestates}, and synchronization errors \eqref{dyn_edge1} is shown in Figure \ref{fig1}--\ref{fig4} for each graph in Figure \ref{graph}, respectively. It is clear that in Figure \ref{fig1}, agents achieve interval bipartite consensus since $\mathcal{G}_1$ is SUB and contains a root node, and in Figures  \ref{fig2}--\ref{fig4}, agents achieve bipartite containment since $\mathcal{G}_2$, $\mathcal{G}_3$, and $\mathcal{G}_4$ contain multiple leader groups. The edge states do not converge to zero and all synchronization errors converge to zero.

\section{Conclusion}\label{section7}
In this paper, we addressed the edge-based synchronization control problem of multi-agent systems interconnected over signed digraphs which contain one or multiple leader groups. First, we presented new spectral properties of signed edge Laplacians containing multiple zero eigenvalues and eigenvectors associated with each zero eigenvalue. Then, we established exponential stability of the synchronization set using the Lyapunov equation-based characterization of the signed edge Laplacians with multiple zero eigenvalues. Finally, we investigated different synchronization scenarios of the signed multi-agent system in edge states, from standard consensus to more complex scenarios of bipartite containment and interval bipartite consensus. Further research aims to extend these results to the case of constrained systems and discrete-time settings.

\section*{References}
\bibliographystyle{IEEEtran}
\bibliography{references_pelin}%,references_angela

\end{document}